\documentclass[runningheads]{llncs}

\usepackage[mathlines]{lineno}
\usepackage{makecell}
\usepackage{multirow}
\usepackage{float}
\usepackage{amssymb}
\usepackage{amsmath, bm}
\usepackage{bm}
\usepackage{algorithm}
\usepackage{algorithmic}
\usepackage{enumerate}
\usepackage{cite}
\usepackage{color}
\usepackage{adjustbox}
\usepackage{blindtext}
\usepackage{marvosym}

\newtheorem{myDef}{Definition}

\makeatletter
\newcounter{myCorollary}
\newenvironment{myCorollary}[1][]{%
  \refstepcounter{myCorollary}%
  \trivlist\item[\hskip\labelsep{\bfseries Corollary~\themyCorollary}]\itshape
  \ifx\\#1\\\else{\bfseries(#1).}\ \fi
}{\endtrivlist}
\makeatother
\usepackage{threeparttable}
\usepackage{booktabs}
\usepackage{longtable}

\usepackage[hidelinks,hypertexnames=false]{hyperref} 
\hypersetup{
colorlinks=true,
linkcolor=red,
citecolor=green
}

\usepackage{xcolor}

\newlength{\ALGEQREFPOS}

\makeatletter
\newcommand{\ALGREF}[1]{%
    \makebox[0pt][l]{%
        \hspace*{\dimexpr\ALGEQREFPOS-\@totalleftmargin\relax}%
        $\triangleright$ #1}}
\makeatother

\usepackage{graphicx}
\usepackage{ulem}
\usepackage{orcidlink}
\usepackage{microtype}
\begin{document}

% The title page has intentionally sparse vertical content.
\raggedbottom

%\linenumbers
%
% \title{Neural Aided Statistical Attack Model for Cryptanalysis     \thanks{Supported by organization x.}}

\title{Extracting CNNs in the Unknown-Architecture and Feedback-Agnostic Setting}
% \title{Breaking Slope and Structure Restrictions: 
% Broadening Cryptanalytic Extraction of PReLU Neural Networks in The Hard-Label Setting}

% Author block aligned with the HotCRP submission form (5 authors).
% Kept commented out: the submission PDF must stay anonymous.
%\author{Jiashuo Liu\inst{1}, Ruijie Ma\inst{2}\orcidlink{0009-0004-7811-5200}, Manman Li\inst{1}, \\
% Yi Chen\inst{3(\mbox{\raisebox{0.05ex}{\scalebox{0.9}{\Letter}}})}\orcidlink{0000-0002-4727-4530}, and Shaozhen Chen\inst{1}}

\author{Jiashuo Liu\inst{1}\orcidlink{0000-0002-9415-5972}, Ruijie Ma\inst{2}\orcidlink{0009-0004-7811-5200}, Manman Li\inst{1}\orcidlink{0000-0001-7625-737X}, Yi Chen\inst{3(\mbox{\raisebox{0.05ex}{\scalebox{0.9}{\Letter}}})}\orcidlink{0000-0002-4727-4530}, Shaozhen Chen\inst{1}}

\institute{Information Engineering University, Zhengzhou, China\\
\email{jiashuoliu@126.com, limanman15@163.com, chenshaozhen@vip.sina.com}\and
Department of Computer Science and Technology, Tsinghua University, Beijing, China\\
\email{marj21@mails.tsinghua.edu.cn}\and
Institute for Advanced Study, Tsinghua University, Beijing, China\\
\email{chenyi2023@tsinghua.edu.cn}}

\maketitle              % typeset the header of the contribution
\begin{abstract}

This paper studies the cryptanalytic extraction of convolutional neural
networks (CNNs). Existing cryptanalytic extraction attacks on CNNs assume
that the network architecture is known, and try to recover model
parameters.
In this paper, we prove for the first time that the architecture assumption
can be removed for CNNs with both max and average pooling. Our core finding
is that the spatial geometry of the weight vectors recovered by existing
parameter-recovery attacks naturally leaks the architecture. We formalize
this geometry and establish its correspondence with the architectural
knowledge of a convolutional layer: (1) The sparsity consistency with the
convolution receptive field reveals the layer type, the kernel size, and
the stride; (2) The numerical consistency with the kernel parameters
reveals the padding mode and the output-channel number; (3) The structural
consistency with the pooling operation reveals the pooling type, the window
size, and the stride. Although the recovered vectors are obtained using
different methods in the raw-output and hard-label settings, their spatial
geometry remains the same. Therefore, our architecture recovery is
feedback-agnostic: combined with a parameter-recovery attack, it yields a
complete cryptanalytic extraction framework that recovers both the
architecture and the parameters in the black-box setting. Extensive
experiments, including both layer-wise and end-to-end ones, on a wide range
of CNNs demonstrate that simultaneously recovering the network architecture
and the model parameters is practical.

\keywords{Cryptanalytic Extraction \and Convolutional Neural Networks \and
Network Architecture Recovery.}
\end{abstract}
%\clearpage
%\flushbottom
%
%
%

% !TeX spellcheck = en_US
% !TEX root = main.tex

\section{Introduction}
\label{sec:introduction}

A trained deep neural network (DNN) embodies expensive data, computation,
and expertise. Consequently, DNNs are typically deployed as a remote black-box
service: users can only submit inputs and observe the returned feedback.
This raises the model extraction problem: can an adversary recover a
network from such interactions? The most challenging goal of model
extraction is functionally equivalent extraction, which requires the
extracted model to return the same raw output as the victim model at every
point of the input space~\cite{DBLP:conf/uss/JagielskiCBKP20}.

The earliest work on model extraction dates to
1991~\cite{DBLP:journals/tnn/Baum91}, and functionally equivalent
extraction was achieved for logistic regression models in
2016~\cite{DBLP:conf/uss/TramerZJRR16} and for ReLU fully connected
networks (FCNs) with 1--2 hidden layers in
2020~\cite{DBLP:conf/uss/JagielskiCBKP20,DBLP:conf/icml/RolnickK20}.

At
Crypto 2020, Carlini et al. argued that model extraction is a cryptanalytic problem in
disguise~\cite{DBLP:conf/crypto/CarliniJM20}. They proposed a differential
extraction attack that applies to ReLU FCNs with an arbitrary number of
hidden layers and achieves functionally equivalent extraction.
This work formally
opened the direction of cryptanalytic extraction of neural networks.

Following this work, the cryptography community has studied FCNs under two feedback
settings: (1) Raw-output setting: the adversary receives the
complete logit vector $\mathcal{F}_{\mathcal{A},\Theta}(x)$; (2) Hard-label setting:
the adversary receives only the predicted label
$z(x)=\arg\max_i\mathcal{F}_{\mathcal{A},\Theta}(x)_i$.
In the raw-output setting, the attack has been made
polynomial-time~\cite{DBLP:conf/eurocrypt/CanalesMartinezCHRSS24}, run
faster in practice~\cite{DBLP:conf/nips/FoersterMSH24}, extended to deep
networks in an end-to-end manner~\cite{DBLP:conf/eurocrypt/LiuSELBP26},
adapted to PReLU and other activation
functions~\cite{DBLP:conf/asiacrypt/ChenDMSWYW25,QLW26,ADF26}, and
examined under finite-precision
evaluation~\cite{cryptoeprint:2026/1943}.
In the hard-label setting, cryptanalytic extraction has been
initiated~\cite{DBLP:conf/asiacrypt/ChenDGSWW24}, made
polynomial-time~\cite{DBLP:conf/eurocrypt/CarliniCHRS25,cryptoeprint:2024/1580}, extended to the
extraction of some deeper
layers~\cite{DBLP:conf/latincrypt/CanalesMartinezS25}, and re-examined for
its query complexity~\cite{DBLP:journals/iacr/ItoMT25}.

Convolutional neural networks
(CNNs) are among the mainstream models for computer vision and image recognition. 
It is natural that a recent and important trend is to extend cryptanalytic
extraction from FCNs to
CNNs.
In the raw-output setting, Chen et al. gave the first algebraic treatment
of max-pooling CNNs and defined three types of points to recover the
convolution parameters~\cite{cryptoeprint:2026/241}; Liu et al. then
extracted deeper VGG-style CNNs layer by layer, using receptive-field
analysis~\cite{hal-05573262}.
In the hard-label setting, Sun et al. handled
average-pooling CNNs by reformulating convolutional layers as
block-Toeplitz matrices~\cite{cryptoeprint:2026/139}; Chen et al.
accelerated the extraction of max-pooling CNNs with the approximate
signature vector method and kernel-centric
clustering~\cite{cryptoeprint:2026/1164}.

A DNN is determined by its architecture $\mathcal{A}$ and parameters
$\Theta$. The architecture $\mathcal{A}$ specifies the computational
structure, namely the layer types, dimensions, and structural
hyperparameters. The parameters $\Theta$ are the trainable numerical values
attached to the structure. All existing attacks on CNNs
assume $\mathcal{A}$ to be known and recover only $\Theta$. This assumption is
unrealistic, since $\mathcal{A}$ is a valuable asset that a realistic
deployment does not disclose.

For an FCN, the unknown architecture information is a list of layer dimensions, and its
recovery is
settled~\cite{DBLP:conf/icml/RolnickK20,DBLP:conf/iclr/DanielyG23,cryptoeprint:2024/1580
,shen2026fcn}.
By contrast, a CNN exploits local receptive fields and weight sharing
to process image data with far fewer parameters.
Its architecture is accordingly richer, comprising the layer type, the
kernel size, the stride, the padding mode, the channel numbers, and the
pooling structure.
Each is both a non-trainable parameter affecting network performance and,
like the numerical parameters, an asset of the model provider.
Because of their diversity and complexity, a CNN's architectural knowledge
is even more valuable to an adversary and harder to obtain. 
This leads to the research question of this paper: 
can the convolutional layers of a CNN be extracted without
the knowledge of network architecture? 
The significance is twofold. (1) Together with existing methods for the fully
connected layers, answering it removes the architectural prior from CNN
cryptanalytic extraction. (2) It makes architecture recovery independent of
parameter recovery, yielding a modular framework that reuses existing
parameter-recovery methods.

A CNN consists of convolutional layers followed
by one or two fully connected layers, for which recovery methods are already
available~\cite{DBLP:conf/icml/RolnickK20,DBLP:conf/iclr/DanielyG23,cryptoeprint:2024/1580
,shen2026fcn}. We focus on the
architecture recovery of the convolutional layers, which remains open. 
We analyze whether the first step of existing parameter-recovery
frameworks, namely point search and equation
solving, carries observational information
that reflects the convolutional architecture. Since all of these methods share
this step~\cite{cryptoeprint:2026/241,hal-05573262,cryptoeprint:2026/139,
cryptoeprint:2026/1164}, our architecture recovery is independent of any
particular method.

\paragraph{\normalfont\normalsize\bfseries{Our Contributions.}}
To the best of our knowledge, we propose the first architecture recovery
method for CNNs. Our key observation is that the spatial structure of the 
recovered vectors directly reveals the CNN architecture, including all the hyperparameters related to convolution and pooling.
Our contributions are summarized as follows.

\begin{itemize}
    \item \textbf{Formal spatial representation.} We give a formal
    representation of the spatial structure of the recovered vectors~(Section~\ref{sec:auxiliary_attack_overview}), built on
    the first step of the parameter-recovery
    framework, namely point search and
    equation
    solving. Since it
    is read directly off the recovered vectors, it is independent of the
    feedback setting.

    \item \textbf{Feedback-agnostic and modular architecture
    recovery.} We decompose the attack process into several subtasks, as
    listed in Table~\ref{tab:architecture_recovery_tasks}: identifying the
    target layer type (Section~\ref{sec:layer}), recovering the
    convolution kernel size and stride
    (Section~\ref{sec:conv_parameter_identification}), the padding mode
    and the output channel number
    (Section~\ref{sec:padding_output_channel_recovery}), and the pooling
    window size and stride
    (Section~\ref{subsec:pooling_structure_identification}).

    The attack is feedback-agnostic and modular. It uses only the
    recovered vectors, which the point search and equation solving that
    open every parameter-recovery pipeline already produce in both the
    raw-output and the hard-label setting, since the critical points
    searched in the raw-output setting and the dual points searched in
    the hard-label setting yield identical recovered vectors. Hence it
    composes with any existing parameter-recovery attack for CNNs.
    Combined with such an attack, it yields a complete
    cryptanalytic extraction framework that recovers each layer's
    architecture before its parameters
    (Figure~\ref{fig:attack_framework_comparison}).

    \item \textbf{Experimental evaluation.} We evaluate the attack on the
    models adopted by existing CNN attacks. In the architecture-unknown
    setting, the architecture of every tested model is recovered
    exactly, in both feedback settings and both upstream modes, with
    only a few additional queries and slightly more running time. To
    show that the attack handles diverse architectural hyperparameters,
    we further design an architecturally diverse CNN and present, task
    by task, the complete evidence of its architecture recovery through
    visualizations (Section~\ref{sec:experiments}). The complete
    experimental data and the source code are released. 
    %at \url{https://anonymous.4open.science/r/conv-architecture-recovery-383}.

\end{itemize}

% !TeX spellcheck = en_US
% !TEX root = ../main.tex

\section{Preliminaries}
\label{sec:preliminaries}

\subsection{Basic Definitions and Notation}
\label{subsec:basic_definitions_notations}

We introduce the basic definitions and notations related to 
convolutional neural networks and cryptanalytic extraction attacks.
\begin{myDef}[($m+n$)-Deep Convolutional Neural Network]
\label{def:cnn}
An $(m+n)$-deep CNN is a map
$\mathcal{F}_{\mathcal{A},\Theta}:\mathbb{R}^{d_0}\rightarrow
\mathbb{R}^{d_{n+1}}$. It consists of
convolutional layers followed by fully connected layers:
\begin{equation*}
    \resizebox{0.98\linewidth}{!}{$\displaystyle
    \mathcal{F}_{\mathcal{A},\Theta}
    =
    \underbrace{
    f^{(n+1)}\circ\sigma^{(n)}\circ\cdots\circ\sigma^{(1)}\circ f^{(1)}
    }_{\mathrm{Fully\ connected\ layers}}
    \circ
    \underbrace{
    \rho^{(m)}\circ\sigma_c^{(m)}\circ f_c^{(m)}
    \circ\cdots\circ
    \rho^{(1)}\circ\sigma_c^{(1)}\circ f_c^{(1)}
    }_{\mathrm{Convolutional\ layers}}
    .$}
\end{equation*}
In the convolutional layer $l$, for $1\leq l\leq m$, $f_c^{(l)}$, $\sigma_c^{(l)}$, and
$\rho^{(l)}$ denote the convolution, ReLU activation, and pooling
operations, respectively.
In the fully connected layer $l$, for $1\leq l\leq n$, $f^{(l)}$ and $\sigma^{(l)}$ denote the affine
transformation and ReLU activation. The final output layer is the affine transformation $f^{(n+1)}$.
\end{myDef}

\begin{myDef}[Model Architecture]
\label{def:architecture}
The model architecture $\mathcal{A}$ is the ordered sequence of layer types, 
dimensions, and structural hyperparameters. Let $L$ denote the index of the
target layer. The architecture information includes
\begin{equation*}
    \mathcal{A}^{(L)}
    =
    \begin{cases}
        (\tau^{(L)},k^{(L)},s_k^{(L)},\pi^{(L)},C_{\mathrm{out}}^{(L)},
        \rho^{(L)},p^{(L)},s_p^{(L)}), & \tau^{(L)}=\mathrm{Conv},\\
        (\tau^{(L)},d_{\mathrm{fc}}^{(L)}), & \tau^{(L)}=\mathrm{FC}.
    \end{cases}
\end{equation*}
The components are specified as follows:
\begin{itemize}
    \item $\tau^{(L)}\in\{\mathrm{Conv},\mathrm{FC}\}$ denotes the target-layer
    type.
    \item If $\tau^{(L)}=\mathrm{Conv}$, the structural hyperparameters include
    kernel size $k^{(L)}$ and convolution stride $s_k^{(L)}$,
    padding mode $\pi^{(L)}\in\{\mathrm{Valid},\mathrm{Same}\}$,
    output-channel number $C_{\mathrm{out}}^{(L)}$, and pooling type
    $\rho^{(L)}\in\{\mathrm{NoPool},\mathrm{AvgPool},\mathrm{MaxPool}\}$;
    when pooling is present, $\mathcal{A}^{(L)}$ includes pooling window size $p^{(L)}$ and
    pooling stride $s_p^{(L)}$.
    \item If $\tau^{(L)}=\mathrm{FC}$, $d_{\mathrm{fc}}^{(L)}$ denotes the number of neurons in that layer.
\end{itemize}
\end{myDef}

As an example, Figure~\ref{fig:conv_layer_example} illustrates a convolutional
layer with input channel number $C^{(L)}=1$ and architecture
    $\mathcal{A}^{(L)}
    =
    (\mathrm{Conv},2,1,\mathrm{Valid},1,\mathrm{MaxPool},2,2)$. Throughout,
$IN^{(L)}\in\mathbb{R}^{C^{(L)}\times H^{(L)}\times W^{(L)}}$ denotes the
input of the convolution operation, where $C^{(L)}$, $H^{(L)}$, and $W^{(L)}$
are its channel number, height, and width, and
$K^{(L)}\in\mathbb{R}^{k^{(L)}\times k^{(L)}}$ denotes the convolution kernel.
\begin{figure}[!htb]
    \centering
    \includegraphics[width=0.98\textwidth]{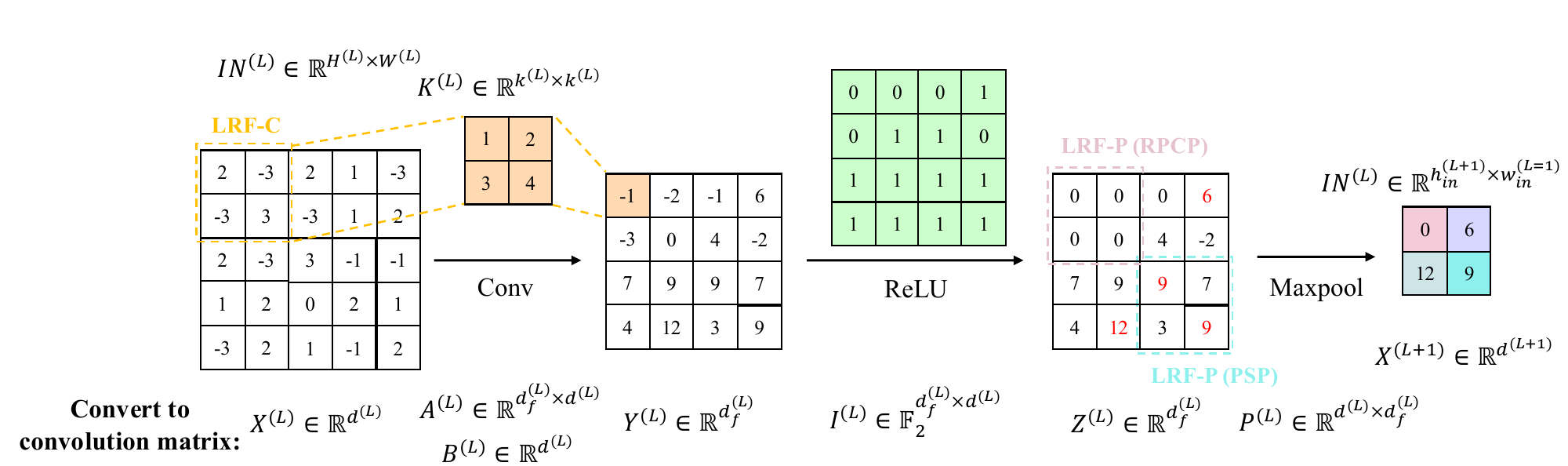}
    \caption{Symbolic view of the example convolutional layer architecture
    $\mathcal{A}^{(L)}=(\mathrm{Conv},2,1,\mathrm{Valid},1,\mathrm{MaxPool},2,2)$.
    The lower labels are the corresponding notations in the convolution-matrix view.
    LRF-C (Local Receptive Field of Convolutional Layer) is the sub-matrix of
    $IN^{(L)}$ multiplied by kernel $K^{(L)}$ to produce the entry
    $\bigl(Y^{(L)}\bigr)_t$, where $0\le t\le d_f^{(L)}-1$.
    LRF-P (Local Receptive Field of Pooling Layer) is the sub-matrix of
    $Z^{(L)}$ processed by pooling to produce the
    entry $\bigl(X^{(L+1)}\bigr)_t$, where $0\le t\le d^{(L+1)}-1$.}
    \label{fig:conv_layer_example}
\end{figure}

\subsection{Parameter Recovery in CNNs}
\label{subsec:preliminaries_signature_recovery}

We review the extraction framework for FCNs and its extension to CNNs, stating
each component only to the extent required by the subsequent analysis. Further
technical details are developed where they are used.

\begin{myDef}[Parameter Extraction Framework~\cite{DBLP:conf/crypto/CarliniJM20}]
\label{def:parameters_extraction_framework}
The parameters are recovered layer by layer.
Let $F_x^{(L-1)}$ denote the recovered prefix and $G_x^{(L+1)}$ the
unrecovered suffix. In a sufficiently small linear neighborhood of $x$, the
target layer is isolated as
$\mathcal{F}_{\mathcal{A},\Theta}(x)
=G_x^{(L+1)} (A^{(L)}F_x^{(L-1)}(x)+B^{(L)}) $.
Denote by $\widehat{A}^{(L)}_i$ the final recovered weight of the $i$-th
neuron in layer $L$. The framework is decomposed into two sequential
phases:

\begin{itemize}
\item \textbf{Neuron Signature Recovery:} 
For each neuron, the attacker performs the following steps:
point collection, equation solving, clustering and sorting, and merging. 
For a ReLU-induced point, the recovered vector $\tilde{A}^{(L)}_i$ is proportional to the neuron signature, with
$\tilde{A}^{(L)}_i \in\{\widehat{A}^{(L)}_i,-\widehat{A}^{(L)}_i\}$ after normalization.
    
\item \textbf{Neuron Sign Recovery:} 
Checking whether $\tilde{A}^{(L)}_i = \widehat{A}^{(L)}_i$ or $\tilde{A}^{(L)}_i = -\widehat{A}^{(L)}_i$.
\end{itemize}
\end{myDef}

At EUROCRYPT 2024, Canales-Mart{\'{\i}}nez et
al.~\cite{DBLP:conf/eurocrypt/CanalesMartinezCHRSS24} stated that the framework in
Definition~\ref{def:parameters_extraction_framework} also works for CNNs
without modification. Each placement of the convolution kernel is regarded as a convolution neuron
whose weight vector is nonzero only on its LRF-C (see Figure~\ref{fig:conv_layer_example}). Stacking these weight vectors yields Definition~\ref{def:convolution_matrix}.

\begin{myDef}[Convolution Matrix~\cite{cryptoeprint:2026/241}]
\label{def:convolution_matrix}
Denote the convolution matrix by $A^{(L)}$. Let
$w_o^{(L)}$ denote the output width of the convolution at layer $L$,
i.e., $\lceil W^{(L)}/s_k^{(L)}\rceil$ under Same padding and
$\lfloor(W^{(L)}-k^{(L)})/s_k^{(L)}\rfloor+1$ under Valid padding. Let
$C^{(L)}=C_{\mathrm{out}}^{(L)}=1$; for $0\leq i\leq d_f^{(L)}-1$, 
define the $i$-th row of the convolution matrix $A^{(L)}$ in two cases.
\begin{itemize}
    \item 
    %\textbf{$\pi^{(L)} = \mathrm{Valid}$:} 

$\pi^{(L)} = \mathrm{Valid}:
A_i^{(L)}=\bigl(\underbrace{0,\ldots,0,K_0^{(L)},0,\ldots,0}_{\left|\overline{A}_{i,1}^{(L)}\right|=W^{(L)}},\ldots,\underbrace{0,\ldots,0,K_{k^{(L)}-1}^{(L)},0,\ldots,0}_{\left|\overline{A}_{i,k^{(L)}}^{(L)}\right|=W^{(L)}}\bigr)
$,

\item
$\pi^{(L)} = \mathrm{Same}$:
\begin{equation*}
\resizebox{\linewidth}{!}{$\displaystyle
A_i^{(L)}=\bigl(\underbrace{0,\ldots,0}_{(W^{(L)}+2t_{\mathrm{pad}}^{(L)})t_{\mathrm{pad}}^{(L)}},\underbrace{0,\ldots,0}_{t_{\mathrm{pad}}^{(L)}},\overline{A}_{i,1}^{(L)},\underbrace{0,\ldots,0}_{t_{\mathrm{pad}}^{(L)}},\ldots,\underbrace{0,\ldots,0}_{t_{\mathrm{pad}}^{(L)}},\overline{A}_{i,k^{(L)}}^{(L)},\underbrace{0,\ldots,0}_{t_{\mathrm{pad}}^{(L)}},\underbrace{0,\ldots,0}_{(W^{(L)}+2t_{\mathrm{pad}}^{(L)})t_{\mathrm{pad}}^{(L)}}\bigr)
$},
\end{equation*}
where
$t_{\mathrm{pad}}^{(L)}=\frac{1}{2}\left(\left(\left\lceil W^{(L)}/s_k^{(L)}\right\rceil-1\right)s_k^{(L)}+k^{(L)}-W^{(L)}\right)$.
\end{itemize}
Let $\mathcal{K}_i^{(L)}$ index the positions of the segments
$\overline{A}_{i,1}^{(L)},\ldots,\overline{A}_{i,k^{(L)}}^{(L)}$
in $A_i^{(L)}:$
\begin{equation*}
\resizebox{\linewidth}{!}{$\displaystyle
    \mathcal{K}_i^{(L)}
    =
    \left\{
    \left(t_{\mathrm{pad}}^{(L)}+\left\lfloor \frac{i}{w_o^{(L)}}\right\rfloor\cdot s_k^{(L)}+j\right)
    \left(W^{(L)}+2t_{\mathrm{pad}}^{(L)}\right)
    +t_{\mathrm{pad}}^{(L)}+\left(i\bmod w_o^{(L)}\right)\cdot s_k^{(L)}+m
    \ \middle|\
    0\le j<k^{(L)},\ 0\le m<k^{(L)}
    \right\}.
$}
\end{equation*}

\end{myDef}

Each entry of the input tensor $\mathrm{IN}^{(L)}$ is indexed by a triple $(q,u,v)$, 
with $q$ the channel index and $(u,v)$ the spatial position. 
Write $\Lambda^{(L)}=\{0,\ldots,H^{(L)}-1\}\times\{0,\ldots,W^{(L)}-1\}$, $\Omega^{(L)}=\{0,\ldots,C^{(L)}-1\}\times\Lambda^{(L)}$ 
for the spatial and the channel–spatial coordinate spaces. 
To write the convolution as a matrix operation, the input tensor $\mathrm{IN}^{(L)}$ is first flattened into the one-dimensional space 
$\mathcal{I}^{(L)}=\{0,\ldots,d^{(L)}-1\}$; the flattened input is denoted by $x^{(L)}$. 
According to the padding mode $\pi^{(L)}$, each spatial dimension is then enlarged by $t_{\mathrm{pad}}^{(L)}$ zero entries on both sides, 
and the resulting tensor is flattened into a vector indexed by the padded space $\widetilde{\mathcal{I}}^{(L)}$. 
The row-major bijection $\psi^{(L)}:\mathcal{I}^{(L)}\to\Omega^{(L)}$, $\psi^{(L)}(i)=(q,u,v)$, identifies $\mathcal{I}^{(L)}$ with $\Omega^{(L)}$. 
Hence the padded indices $\widetilde{\mathcal{I}}^{(L)}$ serve only to write the convolution matrix, 
whereas the entries of $\mathrm{IN}^{(L)}$ live in the unpadded space $\mathcal{I}^{(L)}$. 
The two are related by the padding map $\phi^{(L)}:\mathcal{I}^{(L)}\hookrightarrow\widetilde{\mathcal{I}}^{(L)}$, 
which sends an unpadded index to the padded index of the same position, i.e.\ $(q,u,v)\mapsto(q,u+t_{\mathrm{pad}}^{(L)},v+t_{\mathrm{pad}}^{(L)})$ in coordinates.
For a set $T\subseteq\widetilde{\mathcal{I}}^{(L)}$ we write $(\phi^{(L)})^{-1}(T)=\{j\in\mathcal{I}^{(L)}:\phi^{(L)}(j)\in T\}$ for its preimage, 
which discards the padded positions.

%For example, consider a single-channel $3\times3$ input, a $2\times2$
%kernel $K^{(L)}$, and
%s_k^{(L)}=1$. The convolution matrix $A^{(L)}$ is given by:
%\begin{equation*}
%K^{(L)}=\left(\begin{array}{*{2}{c}}
%1&2\\
%3&4
%\end{array}\right)
%\xrightarrow{\ t_{\mathrm{pad}}^{(L)}=0\ }
%A^{(L)}
%=
%\left(\begin{array}{*{9}{c}}
%1&2&0&3&4&0&0&0&0\\
%0&1&2&0&3&4&0&0&0\\
%0&0&0&1&2&0&3&4&0\\
%0&0&0&0&1&2&0&3&4
%\end{array}\right).
%\end{equation*}

Based on the convolution matrix, Sun et al.~\cite{cryptoeprint:2026/139} apply the
parameter extraction framework to ReLU CNNs with average pooling. However, Chen et al.~\cite{cryptoeprint:2026/241} pointed out that the framework is not directly applicable to CNNs with max pooling, 
because the selection of the maximum value in LRF-P introduces additional nonlinearity. 
By jointly accounting for ReLU activations and max-pooling selections, 
prior work extends the critical point~\cite{cryptoeprint:2026/241} and the dual point~\cite{cryptoeprint:2026/1164}.

\begin{myDef}[Critical Points~\cite{cryptoeprint:2026/241}]
\label{def:critical_points}
A CNN has three critical point types:
\begin{itemize}
    \item A ReLU-pooling critical point (RPCP) is an input $x$ for which the
    input to the $i$-th ReLU vanishes in convolutional layer $L$,
    $A_i^{(L)}x^{(L)}+\bigl(B^{(L)}\bigr)_i=0$; this value is termed a
    convolution response. With max pooling, this response must also
    be the maximum in its LRF-P.
    \item A pooling switching point (PSP) is an input $x$ at which two
    responses in one LRF-P tie for the maximum in convolutional layer $L$:
        $A_{r_1}^{(L)}x^{(L)}+\bigl(B^{(L)}\bigr)_{r_1}
        =A_{r_2}^{(L)}x^{(L)}+\bigl(B^{(L)}\bigr)_{r_2}$.
    \item A fully connected critical point (FCP) is an input $x$ for which
    the input to the $i$-th ReLU vanishes in fully connected layer $L$,
    $A_i^{(L)}x^{(L)}+\bigl(B^{(L)}\bigr)_i=0$.
\end{itemize}
LRF-P instances of the RPCP and PSP are shown in Figure~\ref{fig:conv_layer_example}.
\end{myDef}

\begin{myDef}[Decision-Boundary Point and Dual Point~\cite{cryptoeprint:2026/1164}]
\label{def:dual_points}
A decision-boundary point is an input $x$ where the two largest logits
coincide, i.e.\ for some $i\neq j$:
\[
\mathcal{F}_{\mathcal{A},\Theta}(x)_i
=\mathcal{F}_{\mathcal{A},\Theta}(x)_j
=\max \mathcal{F}_{\mathcal{A},\Theta}(x).
\]
A dual point is both a decision-boundary point and a critical point.
Since critical points are of three types, RPCP, PSP, and FCP, so are dual points:
dual RPCP, dual PSP, and dual FCP.

\end{myDef}

For each critical or dual point $x$, the attacker samples small perturbations $\vec{\delta}$
and computes $\mathcal{H}(x;\vec{\delta})=\mathcal{F}_{\mathcal{A},\Theta}(x+\vec{\delta})
+\mathcal{F}_{\mathcal{A},\Theta}(x-\vec{\delta})-2\mathcal{F}_{\mathcal{A},\Theta}(x)$.
These values form a linear system, which is solved to obtain the recovered 
vector $v_x$~\cite{DBLP:conf/crypto/CarliniJM20}.
\begin{lemma}[Recovered Vectors in CNNs~\cite{cryptoeprint:2026/241,cryptoeprint:2026/1164,hal-05573262}]
\label{lem:recovery_vector_cnn}
For a critical or dual point $x$, there exists a nonzero scalar $g_x$
such that the recovered vector $v_x$ reads
\begin{equation}
    \label{eq:recovered_signature_vector}
    v_x=
    \begin{cases}
        g_x A_{m,i}^{(L)}[\mathcal{K}_i^{(L)}],
        & x\text{ is a (dual) RPCP at layer }L,\\
        g_x\bigl(A_{m,i}^{(L)}-A_{m,j}^{(L)}\bigr)[\mathcal{K}_i^{(L)}\cup\mathcal{K}_j^{(L)}],
        & x\text{ is a (dual) PSP at layer }L,\\
        g_x A_i^{(L)},
        & x\text{ is a (dual) FCP at layer }L,
    \end{cases}
\end{equation}
where $m$ denotes the output-channel index.
\end{lemma}

For a (dual) RPCP or (dual) FCP, the recovered vector has the same form as in
an FCN. Hence, the parameter extraction framework in Definition~\ref{def:parameters_extraction_framework} directly recovers
the neuron signature up to a nonzero scalar.

For a (dual) PSP, the support of $v_x$ selects exactly two rows of
$A_m^{(L)}$; let these be the rows $A_{m,i}^{(L)}$ and $A_{m,j}^{(L)}$
associated with the spatial positions $i<j$. The
nonzero entries of the two rows are shifted copies of the same kernel.
The two input patches at the nonzero positions of $A_{m,i}^{(L)}$ and
$A_{m,j}^{(L)}$ have equal inner products with the convolution kernel. 
Collecting all such equalities yields a linear system in $K^{(L)}$~\cite{cryptoeprint:2026/241}.

The neuron sign recovery method in the framework can be adapted to CNNs. Moreover, the
convolutional structure admits more specialized methods, see~\cite{cryptoeprint:2026/241,cryptoeprint:2026/1164,hal-05573262} for details.

Since critical points and dual points yield the same recovered
vectors, the following discussion uses critical points as illustrative examples.

% !TeX spellcheck = en_US
% !TEX root = ../main.tex

\section{Auxiliary Concepts of Recovered Vectors in CNNs}
\label{sec:auxiliary_attack_overview}

Based on the correspondence between the convolution kernel and the row
vectors of the convolution matrix (Definition~\ref{def:convolution_matrix}),
this section describes the spatial properties of the recovered vector and formalizes
the auxiliary definitions used in the subsequent analysis.

\subsection{Structural Properties of Recovered Vectors}

According to Eq.~\eqref{eq:recovered_signature_vector}, the recovered vector
$v_x$ is either one row of the convolution matrix or the difference of two
rows. The kernel entries
in each row are confined to the corresponding LRF-C,
while all entries outside the LRF-C are zero.
We introduce two sets, $S_x$ and $\mathcal{W}_i^{(L)}$,
to characterize the relationship between recovered vectors and the convolution
matrix.

\begin{myDef}[Effective Position Set of a Recovered Vector]
\label{def:support_channel_slices}
For a recovered vector $v_x$, $S_x$ is defined as
$S_x=\{i\in\mathcal{I}^{(L)}:\ (v_x)_i\ne0\}$.
\end{myDef}

\begin{myDef}[Effective Position Set of a Convolution-Matrix Row]
\label{def:convolution_row_support}
For the $i$-th row $A_{m,i}^{(L)}$ of the $m$-th output channel,
$\mathcal{W}_i^{(L)}\subseteq\mathcal{I}^{(L)}$ is defined as
    $\mathcal{W}_i^{(L)}
    =(\phi^{(L)})^{-1}\bigl(
    \{j\in\widetilde{\mathcal{I}}^{(L)}:(A_{m,i}^{(L)})_j\ne0\}
    \cap\mathcal{K}_i^{(L)}
    \bigr)$,
that is, the input positions of the $i$-th LRF-C; it is independent of $m$,
since all output channels share the same LRF-Cs.
\end{myDef}

\begin{proposition}[Shared Sparsity Pattern per LRF-C]
\label{prop:recovery_vector_sparsity_consistency}
For a target-layer RPCP, PSP, or FCP, the following equation holds:
\begin{equation*}
    \label{eq:point_type_support}
    S_x=
    \begin{cases}
        \mathcal{W}_{i}^{(L)},
        & x\text{ is an RPCP associated with the }i\text{-th LRF-C},\\
        \mathcal{W}_{i}^{(L)}\cup\mathcal{W}_{j}^{(L)},
        & x\text{ is a PSP associated with the }i\text{-th and }j\text{-th LRF-Cs},\\
        \mathcal{I}^{(L)},
        & x\text{ is an FCP}.
    \end{cases}
\end{equation*}
Consequently, RPCPs associated with the same LRF-C, as well as PSPs
associated with the same pair of LRF-Cs, have the same sparsity pattern.

\end{proposition}

\begin{proof}
By Eq.~\eqref{eq:recovered_signature_vector}, $v_x$ equals $g_x$ times
the restricted row (resp. row difference) and is supported on
$\mathcal{K}_i^{(L)}$ (resp. $\mathcal{K}_i^{(L)}\cup\mathcal{K}_j^{(L)}$),
so restriction preserves the support: $S_x=\mathcal{W}_i^{(L)}$ for an
RPCP. For a PSP, the two rows are shifted copies of the same kernel, so
for a generic kernel no cancellation occurs on the overlap and
$S_x=\mathcal{W}_i^{(L)}\cup\mathcal{W}_j^{(L)}$. The FCP case is
immediate, since a fully connected row is dense and $S_x=\mathcal{I}^{(L)}$.
\end{proof}

\begin{proposition}[Proportional Recovered Vectors per LRF-C]
\label{prop:recovery_vector_numerical_consistency}
For a target-layer RPCP, PSP, or FCP, the following equation holds for every
$t\in S_x$:
\begin{equation*}
    (v_x)_t=
    \begin{cases}
        g_x\bigl(A_{m,i}^{(L)}\bigr)_{\phi^{(L)}(t)},
        & x\text{ is an RPCP},\\
        g_x\bigl(A_{m,i}^{(L)}-A_{m,j}^{(L)}\bigr)_{\phi^{(L)}(t)},
        & x\text{ is a PSP},\\
        g_x\bigl(A_i^{(L)}\bigr)_{t},
        & x\text{ is an FCP}.
    \end{cases}
\end{equation*}
\end{proposition}

\begin{proof}
On the effective positions, Eq.~\eqref{eq:recovered_signature_vector}
expresses $v_x[S_x]$ as $g_x$ times the restricted row (resp. row
difference) identified in
Proposition~\ref{prop:recovery_vector_sparsity_consistency}, and the
scalar $g_x$ depends only on $x$; hence two points associated with the
same LRF-C (resp. the same pair of LRF-Cs) have proportional recovered
vectors.
\end{proof}

For each input channel $q\in\{0,\ldots,C^{(L)}-1\}$, let $S_x^{[q]}$ denote the subset
of $S_x$ located in the $q$-th input-channel block.

\begin{proposition}[Equal Channel Slices of the Support]
\label{prop:input_channel_spatial_consistency}
For a target-layer RPCP or PSP, the following equation holds:
\begin{equation*}
    S_x^{[q]}=S_x^{[q']},
    q,q'\in\{0,\ldots,C^{(L)}-1\},\ q\ne q',\quad
    S_x^{[q]}\ne\emptyset,\ S_x^{[q']}\ne\emptyset.
\end{equation*}
\end{proposition}

\begin{proof}
By Eq.~\eqref{eq:recovered_signature_vector},
RPCP rows and PSP row differences repeat
the same spatial kernel pattern in every input-channel block. Hence, all
input-channel slices have identical effective position sets.
\end{proof}

By Proposition~\ref{prop:input_channel_spatial_consistency}, multiple input
channels do not affect the spatial sparsity pattern of recovered vectors.
For clarity, the following discussion assumes a single input channel: 
$\Omega^{(L)}$ = $\Lambda^{(L)}$.

\subsection{Auxiliary Definitions}

To characterize the spatial properties of recovered vectors in the following
analysis, we introduce several auxiliary definitions. We first classify PSPs
into three types according to the relative positions of their corresponding LRF-Cs.

\begin{myDef}[PSP Types]
\label{def:psp_types}
For a target-layer PSP, let
$(r_1,c_1)$ and $(r_2,c_2)$ be the starts of the two associated LRF-Cs in $\Lambda^{(L)}$.
They satisfy $\psi^{(L)}\bigl(\min\mathcal{W}_{i}^{(L)}\bigr)=(r_1,c_1)$ and
$\psi^{(L)}\bigl(\min\mathcal{W}_{j}^{(L)}\bigr)=(r_2,c_2)$.
The positions define three PSP types:
\begin{itemize}
\item A horizontal PSP satisfies $r_1=r_2$ and $c_1\ne c_2$; its
two LRF-Cs' starts lie on the same row.
\item A vertical PSP satisfies $r_1\ne r_2$ and $c_1=c_2$; its
two LRF-Cs' starts lie on the same column.
\item An off-axis PSP satisfies $r_1\ne r_2$ and $c_1\ne c_2$;
its two LRF-Cs' starts differ in both spatial directions.
\end{itemize}
\end{myDef}

To further describe the spatial geometry of recovered vectors, we formalize the
smallest axis-aligned box containing their effective positions and the binary
pattern within that box.

\begin{myDef}[Spatial Box and Shape of a Recovered Vector]
\label{def:spatial_box_shape}
For a recovered vector $v_x$, define its spatial bounding box, denoted
by $\operatorname{Box}_x$, as
\begin{equation*}
\label{eq:architecture_spatial_box}
\operatorname{Box}_x
=
[u_{\min},u_{\max}]
\times
[v_{\min},v_{\max}],
\end{equation*}
where
$u_{\min}=\min_{s\in S_x}[\psi^{(L)}(s)]_u$
and
$u_{\max}=\max_{s\in S_x}[\psi^{(L)}(s)]_u$,
with $v_{\min}$ and $v_{\max}$ defined in the same way.
The corresponding shape is
$\operatorname{Shape}_x
=
(h_x,w_x)$, where
$h_x=u_{\max}-u_{\min}+1$,
$w_x=v_{\max}-v_{\min}+1$.
\end{myDef}

\begin{myDef}[Effective Mask of a Recovered Vector]
\label{def:support_mask}
For a recovered vector $v_x$, the effective mask is a binary matrix
$\operatorname{Mask}_x\in\{0,1\}^{h_x\times w_x}$.
Let $\psi^{(L)}(S_x)$ denote the element-wise spatial mapping of $S_x$. The
$(i,j)$-th element is expressed as
\begin{equation*}
\left(\operatorname{Mask}_x\right)_{i,j}
=
    \left[
    (u_{\min}+i,v_{\min}+j)\in\psi^{(L)}(S_x)
    \right],
    \;
    0\leq i<h_x,
    \;
    0\leq j<w_x.
\end{equation*}
\end{myDef}

To understand how Definitions~\ref{def:psp_types}--\ref{def:support_mask}
work, Figure~\ref{fig:evidence} illustrates an RPCP
and the three PSP types, together with graphical representations of their
respective LRF-Cs, Boxes, and Masks.
\begin{figure}[!htb]
\centering
\begin{minipage}[t]{0.24\textwidth}
\centering
\includegraphics[width=\linewidth]{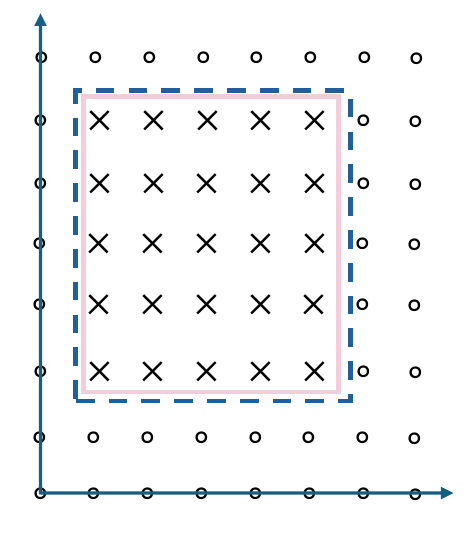}\\[-2pt]
{\small (a) RPCP}
\end{minipage}\hfill
\begin{minipage}[t]{0.24\textwidth}
\centering
\includegraphics[width=\linewidth]{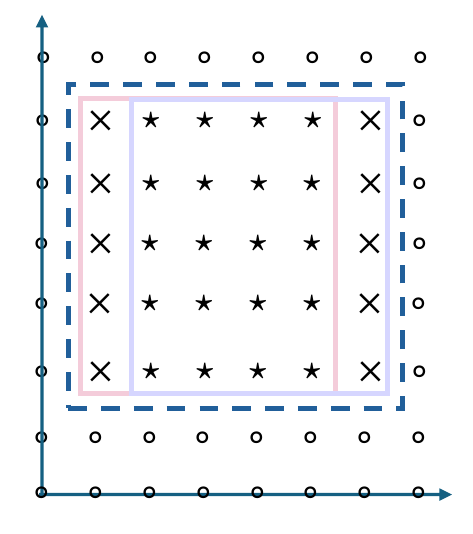}\\[-2pt]
{\small (b) Horizontal PSP}
\end{minipage}\hfill
\begin{minipage}[t]{0.24\textwidth}
\centering
\includegraphics[width=\linewidth]{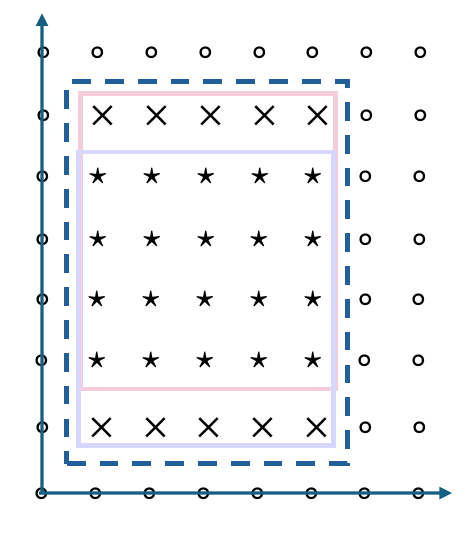}\\[-2pt]
{\small (c) Vertical PSP}
\end{minipage}\hfill
\begin{minipage}[t]{0.24\textwidth}
\centering
\includegraphics[width=\linewidth]{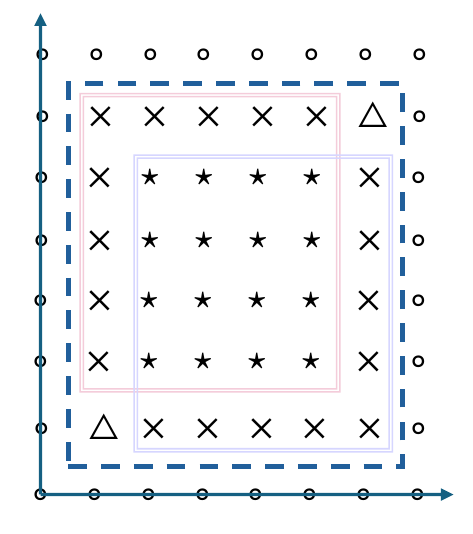}\\[-2pt]
{\small (d) Off-axis PSP}
\end{minipage}
\caption{Examples of the four point types under the architecture
$\mathcal{A}^{(L)}=(\mathrm{Conv},5,1,\mathrm{Valid},1,\mathrm{MaxPool},2,2)$.
Solid lines denote LRF-Cs, and dashed lines denote the Boxes.
Positions marked with $\times$ and $\star$ are effective positions in
$\operatorname{Mask}_x$, where $\star$ further denotes the overlapping
positions of the two LRF-Cs; a
position marked with $\circ$ lies outside the corresponding Box and satisfies
$v_x=0$; and $\triangle$ marks a position that
lies inside the corresponding Box but satisfies $v_x=0$.}
\label{fig:evidence}
\end{figure}

% !TeX spellcheck = en_US
% !TEX root = ../main.tex
\section{Our Cryptanalytic Extraction Framework for CNNs}
\label{sec:tasks}

\subsection{Adversarial Goals and Assumptions}
\label{subsec:attack_setting_goal}

A cryptanalytic extraction attack involves an adversary and an
oracle: the adversary submits generated inputs to the oracle and exploits the
feedback to recover a model functionally equivalent to the victim
model~\cite{DBLP:conf/crypto/CarliniJM20}. Existing cryptanalytic extraction
attacks on
CNNs~\cite{cryptoeprint:2026/139,hal-05573262,cryptoeprint:2026/241,cryptoeprint:2026/1164}
recover only the parameters $\Theta$, assuming that the architecture
$\mathcal{A}$ is known in advance. However, such
architectural information is typically unavailable to the adversary.

\subsubsection{Goals.}
We propose a complete cryptanalytic extraction
framework for CNNs with two components: architecture recovery and
parameter recovery.
The key is to achieve the architecture-recovery
attack and tailor it to be directly compatible with the parameter
extraction framework (Definition~\ref{def:parameters_extraction_framework}),
adapting to existing parameter recovery methods. Within the architecture-recovery attack, guided by the
architecture knowledge in
Definition~\ref{def:architecture},
we decompose the goal into the tasks summarized in Table~\ref{tab:architecture_recovery_tasks}.

\begin{table}[!htb]
\centering
\begin{threeparttable}
\caption{Task decomposition for architecture recovery.}
\label{tab:architecture_recovery_tasks}
\footnotesize
\setlength{\tabcolsep}{3pt}
\begin{tabular}{p{0.08\textwidth}p{0.69\textwidth}p{0.17\textwidth}}
    \toprule
    Task & Description & Section \\
    \midrule
    1 & Identify whether the target layer is convolutional or fully connected,
    yielding $\widehat{\tau}^{(L)}\in\{\mathrm{Conv},\mathrm{FC}\}$. &
    \S\ref{sec:layer} \\
    \addlinespace
    \multicolumn{3}{l}{\textit{Convolutional-layer architecture recovery}
    $(\widehat{\tau}^{(L)}=\mathrm{Conv})$} \\
    \cmidrule(lr){1-3}
    2 & Recover the convolution kernel size $\widehat{k}$. &
    \S\ref{sec:conv_parameter_identification} \\
    3 & Recover the convolution stride $\widehat{s}_k$. &
    \S\ref{sec:conv_parameter_identification} \\
    4 & Recover the padding mode
    $\widehat{\pi}\in\{\mathrm{Valid},\mathrm{Same}\}$. &
    \S\ref{sec:padding_output_channel_recovery} \\
    5 & Recover the output-channel number $\widehat{C}_{\mathrm{out}}$. &
    \S\ref{sec:padding_output_channel_recovery} \\
    6 & Recover the pooling type
    $\widehat{\rho}\in\{\mathrm{NoPool},\mathrm{AvgPool},\mathrm{MaxPool}\}$,
    window size $\widehat{p}$, and stride $\widehat{s}_p$. &
    \S\ref{subsec:pooling_structure_identification} \\
    \addlinespace
    \multicolumn{3}{l}{\textit{Fully-connected-layer architecture recovery}
    $(\widehat{\tau}^{(L)}=\mathrm{FC})$} \\
    \cmidrule(lr){1-3}
    7 & Recover the layer width $\widehat{d}_{\mathrm{fc}}^{(L)}$. & \cite{DBLP:conf/icml/RolnickK20,DBLP:conf/iclr/DanielyG23,cryptoeprint:2024/1580
,shen2026fcn} \\
    \bottomrule
\end{tabular}
\begin{tablenotes}[flushleft]
    \footnotesize
    \item \textit{Note.} After identifying the layer type in Task~1, the
    attack follows either the convolutional branch (Tasks~2--6) or the fully
    connected branch (Task~7).
\end{tablenotes}
\end{threeparttable}
\end{table}

\subsubsection{Assumptions.} We make the following assumptions about the oracle and the capabilities of
the adversary.
\begin{itemize}
    \item \textbf{Full-domain input.}
    The adversary can query arbitrary inputs $x\in\mathbb{R}^{d_0}$.
    \item \textbf{Precise computation.}
    The network is specified and evaluated using 64-bit floating-point arithmetic.
    \item \textbf{ReLU activation.}
    All activation functions are ReLUs.
\end{itemize}
In our setting, the adversary knows only the input dimension $d_0$ and the
output dimension $d_{n+1}$; the types and hyperparameters of the hidden
layers (Definition~\ref{def:architecture}) remain unknown and are to be
recovered.

\subsection{Complete Cryptanalytic Extraction Framework}
\label{subsec:attack_framework}

\subsubsection{New Framework.}

Figure~\ref{fig:attack_framework_comparison} compares the original CNN
parameter-recovery framework with our complete cryptanalytic extraction
framework constructed in the architecture-unknown scenario.

\begin{figure}[!htb]
    \centering
    \includegraphics[width=0.95\textwidth]{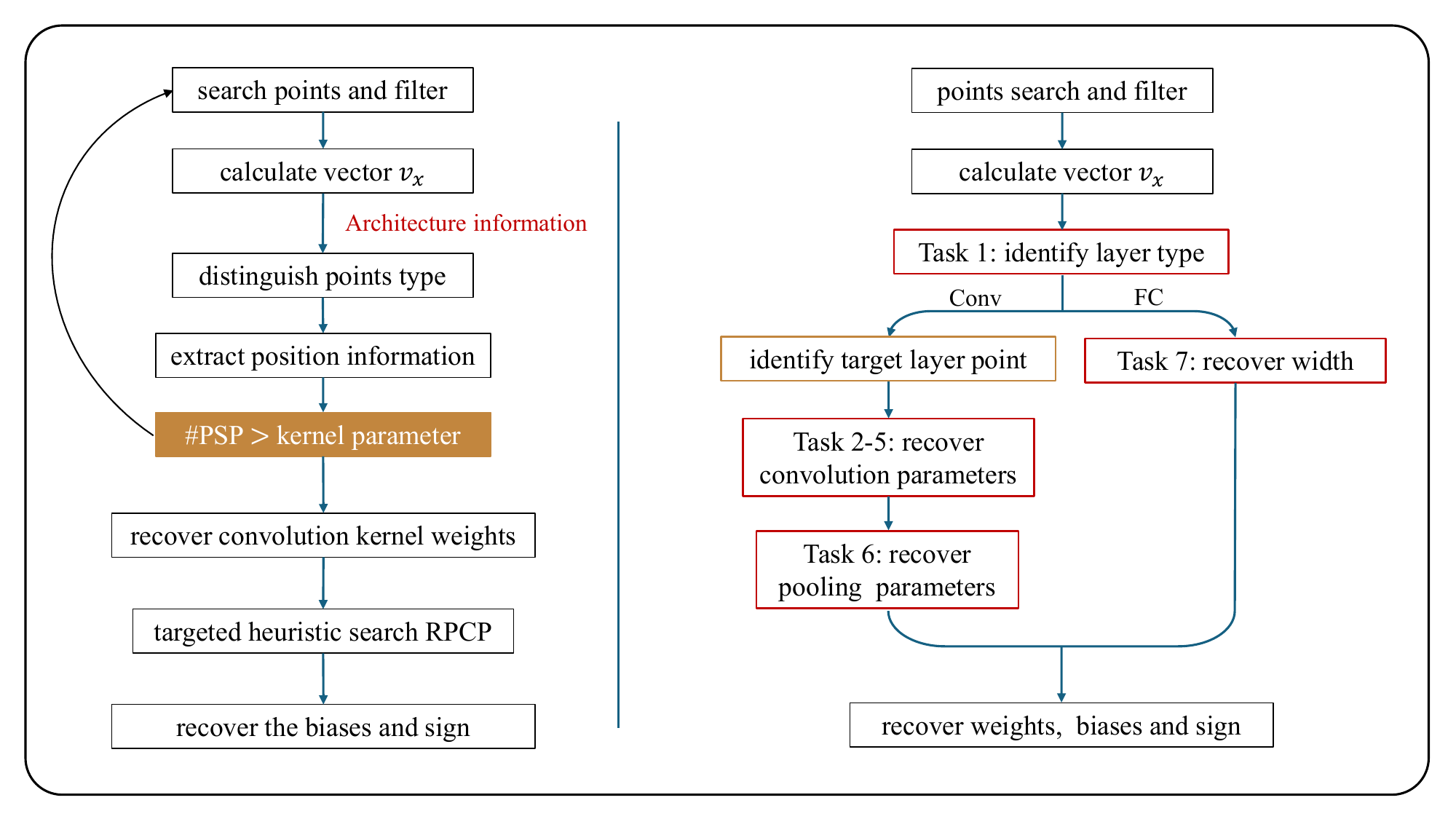}
    \caption{Left: Original CNN parameter-recovery framework under the architecture-known scenario. Right: Our cryptanalytic extraction framework
    under the architecture-unknown scenario.}
    \label{fig:attack_framework_comparison}
\end{figure}

In our complete framework, the attack operates within a single target
layer. For this layer, the first two steps are shared with its parameter
recovery, and their procedure is decided by the feedback mode of the model.
A CNN has two types of layers, convolutional layers and fully connected
layers. We first identify the layer type.
If the current layer is fully connected, its width is recovered by existing
methods.
If it is convolutional, we recover its architecture hyperparameters one by
one in Tasks~2 to 6.
After the architecture recovery is completed, we can connect to the
existing parameter-recovery methods for CNNs.

\subsubsection{Core Idea of the Architecture Recovery.}
Our attack exploits the geometric and numerical relationships among the
convolution kernel, the convolution matrix and the recovered vectors.
A recovered vector is a row of the convolution matrix, or the difference of two
rows (Eq.~\eqref{eq:recovered_signature_vector}), and every row of the
convolution matrix has support only on its corresponding LRF-C.

Convolutional and fully connected layers differ in the types of points they
produce. Therefore, the existence of RPCPs and PSPs is a sufficient condition
for identifying the target layer as convolutional.

Consider the recovered vector of an RPCP. Its effective position set
$S_x=\mathcal{W}_i^{(L)}$ is a single LRF-C of size
$k^{(L)}\times k^{(L)}$, so the shape of its Box reveals the kernel size
$k^{(L)}$. It also follows that the offsets between the Box starting points of
any two such points are integer multiples of the stride $s_k^{(L)}$ within
$\Omega^{(L)}$, and sufficiently dense sampling recovers $s_k^{(L)}$.

Same padding adds zeros around the input and enlarges the effective input
space, so the LRF-Cs at the boundaries extend into the padding area, the
effective position sets of the boundary recovered vectors shrink, and their
Boxes are truncated. Valid padding leaves no such area and no Box is
truncated, and this difference in the boundary Boxes distinguishes the two
padding modes.

A PSP originates from a switch of the max-pooling selection, and it therefore
indicates that the target layer uses max pooling. The max-pooling parameters
are recovered following the same idea as the convolution-kernel parameters.

Average pooling spreads every value of an LRF-P over the $p\times p$ window,
so each coefficient of the recovered vector is spread evenly over that window,
and the extent of the spread gives $p$. When layer $L$ contains no pooling, the recovered vector retains
the sparse convolution structure, so no such repeated-block pattern appears.

\subsubsection{Complexity of the Architecture Recovery.}
The number of points required by our architecture recovery is determined by
the identification tasks, not by the parameter count of the
target layer. In the existing parameter-recovery framework for CNNs, the
recovered vectors are stacked into a linear system that is solved for the
kernel of each output channel. For every output channel, the number
of PSPs must exceed the number of parameters of that channel's kernel $(k^{(L)})^{2}$. 
In contrast, the architecture recovery exploits only the
qualitative signatures of the recovered vectors, and it suffices to find
the minimal cluster of PSPs or RPCPs whose numerical ratios concentrate on a single candidate.
Therefore, a few RPCPs and PSPs are enough for all identification tasks, and the
additional query cost is a constant independent of
$(k^{(L)})^{2}$.
% !TeX spellcheck = en_US
% !TEX root = ../main.tex

\section{Identifying the Target Layer Type}
\label{sec:layer}

In this section, our task is to identify whether the target layer is
convolutional or fully connected, yielding
$\widehat{\tau}^{(L)}\in\{\mathrm{Conv},\mathrm{FC}\}$. Following the
layer-by-layer paradigm of the cryptanalytic extraction framework, let
$\mathcal{P}^{(L)}$ denote the set of candidate points retained for layer $L$ after
filtering out the points attributed to the recovered shallower layers, and
let $\mathcal{P}_{\mathrm{RPCP}}^{(L)}$, $\mathcal{P}_{\mathrm{PSP}}^{(L)}$, and
$\mathcal{P}_{\mathrm{FCP}}^{(L)}$ denote the RPCP, PSP, and FCP subsets of
the points that actually arise from the target layer $L$.
According to the origins of the points in Definition~\ref{def:critical_points}, we propose
Lemma~\ref{lem:layer_point_composition} to characterize the relationship
between the layer type and the point types.

\begin{lemma}[Layer Type Characterized by Point Existence]
\label{lem:layer_point_composition}
Under the assumed CNN ordering in Definition~\ref{def:cnn}, the existence
of an RPCP or PSP in the target layer or deeper layers indicates
$\widehat{\tau}^{(L)}=\mathrm{Conv}$:
\begin{equation*}
    \label{eq:layer_type_iff_conv}
    \mathcal{P}^{(L)}_{\mathrm{RPCP}}
    \cup
    \mathcal{P}^{(L)}_{\mathrm{PSP}}
    \cup
    \bigcup_{L'>L}
    \left(
    \mathcal{P}^{(L')}_{\mathrm{RPCP}}
    \cup
    \mathcal{P}^{(L')}_{\mathrm{PSP}}
    \right)
    \neq\emptyset
    \;\Longrightarrow\;
    \widehat{\tau}^{(L)}=\mathrm{Conv}.
\end{equation*}
\end{lemma}

\begin{proof}
All convolutional layers precede the fully connected layers. If $L$ is fully
connected, both $L$ and every deeper hidden layer are fully connected. Hence
only fully connected layers can occur:
    $\mathcal{P}^{(L)}
    =
    \mathcal{P}^{(L)}_{\mathrm{FCP}}
    \cup
    \bigcup_{L'>L}
    \mathcal{P}^{(L')}_{\mathrm{FCP}}$.
In conclusion, an observed RPCP or PSP in the target layer or deeper layers
excludes the fully connected case, yielding $\widehat{\tau}^{(L)}=\mathrm{Conv}$.
\end{proof}

Lemma~\ref{lem:layer_point_composition} transforms Task~1 into identifying the point type.
In the following, we study how to identify RPCPs, PSPs, and FCPs arising
from the target layer or deeper layers.

\subsection{Identifying Point Type}
\label{subsec:fcp_identification}

We separately discuss the spatial properties of the recovered vectors
obtained in the target-layer input for the six kinds of points: RPCPs,
PSPs, and FCPs, each arising from the target layer or from a deeper
layer.

\subsubsection{Identifying RPCPs.}
\label{subsec:rpcp_psp_identification}

The spatial characteristics of a target-layer RPCP's recovered vector
are given by Proposition~\ref{prop:recovery_vector_sparsity_consistency}.
We formalize them as Corollary~\ref{prop:target_layer_rpcp}.

\begin{myCorollary}[Box, Shape, and Mask of Target-Layer RPCPs]
\label{prop:target_layer_rpcp}
For a target-layer RPCP $x$, the
following hold:
$
    \operatorname{Box}_x=[r_x,r_x+k^{(L)}-1]\times[c_x,c_x+k^{(L)}-1],
    \operatorname{Shape}_x=(k^{(L)},k^{(L)}),
    \operatorname{Mask}_x=\mathbf{1}_{k^{(L)}\times k^{(L)}}.
$
\end{myCorollary}

For an RPCP $x$ arising from a deeper convolutional layer $L'>L$, the
LRF-C of $x$ at layer $L'$ projects onto the target-layer input
space $\Lambda^{(L)}$, so the spatial projection $\psi^{(L)}(S_x)$ of its
support set becomes a union of LRF-Cs of layer $L$.

\begin{proposition}[Deeper-Layer RPCP Supports Are Unions of LRF-Cs]
\label{prop:deeper_layer_rpcp}
For an RPCP $x$ arising from a deeper convolutional layer $L'>L$,
    $\psi^{(L)}(S_x)
    =
    \bigcup_{m\in\mathcal{M}_x}
    \bigl(
    [r_m,r_m+k^{(L)}-1]\times[c_m,c_m+k^{(L)}-1]
    \bigr)$,
where $\mathcal{M}_x$ indexes the LRF-Cs of layer $L$ whose responses
contribute to $x$. Consequently,
\begin{equation*}
\begin{gathered}
    \operatorname{Box}_x
    =
    [\min_{m} r_m,\ \max_{m} r_m+k^{(L)}-1]
    \times
    [\min_{m} c_m,\ \max_{m} c_m+k^{(L)}-1],\\
    \operatorname{Shape}_x
    =
    (\max_{m} r_m-\min_{m} r_m+k^{(L)},\ \max_{m} c_m-\min_{m} c_m+k^{(L)}),
\end{gathered}
\end{equation*}
and $\operatorname{Mask}_x$ indicates the positions covered by the union
within $\operatorname{Box}_x$.
\end{proposition}

\begin{proof}
We prove by recursion on the layer distance $L'-L$.
\begin{enumerate}
    \item \textbf{Base case ($L'=L+1$).} Each convolution response of
    layer $L+1$ is computed from the $p^{(L)}\times p^{(L)}$ convolution responses
    of layer $L$ within one LRF-P of layer $L$ (pooling window $p^{(L)}$,
    stride $s_p^{(L)}$). 
    Under AvgPool, every response of layer $L$ within the LRF-P is propagated,
    and the overlapping LRF-Cs of the propagated responses fill the
    enclosing rectangle.
    Under MaxPool, only the selected maximum
    responses are propagated, and the uncovered positions may remain. In
    both cases, $S_x$ is a union of LRF-Cs of layer $L$. The NoPool case is analogous to the AvgPool case
    but simpler.
    \item \textbf{Recursive step ($L'>L+1$).} By the recursive hypothesis,
    $S_x$ in the input coordinates of layer $L'-1$ is a union
    of LRF-Cs of layer $L'-1$. Each LRF-C of layer $L'-1$ in turn projects
    onto layer $L'-2$ as a union of LRF-Cs of layer $L'-2$ by the
    one-layer argument above. Repeating this projection until the
    target-layer input is reached, $S_x$ is a union of LRF-Cs of layer
    $L$.
\end{enumerate}

In both cases, $\operatorname{Box}_x$ is the smallest rectangle enclosing
these LRF-Cs and $\operatorname{Mask}_x$ indicates the covered positions,
which gives the stated structure.
\end{proof}

\subsubsection{Identifying PSPs.}

The spatial characteristics of a PSP's recovered vector
are given by Proposition~\ref{prop:recovery_vector_sparsity_consistency}.
We formalize them as
Corollary~\ref{prop:target_layer_psp}.

\begin{myCorollary}[Box, Shape, and Mask of Target-Layer PSPs]
\label{prop:target_layer_psp}
For a target-layer PSP $x$ with
$\psi^{(L)}\bigl(\min\mathcal{W}_{i}^{(L)}\bigr)=(r_1,c_1)$,
$\psi^{(L)}\bigl(\min\mathcal{W}_{j}^{(L)}\bigr)=(r_2,c_2)$, the following hold:
\begin{equation*}
\resizebox{0.92\linewidth}{!}{$\displaystyle
\begin{gathered}
    \operatorname{Box}_x
    =
    [\min\{r_1,r_2\},\max\{r_1,r_2\}+k^{(L)}-1]
    \times
    [\min\{c_1,c_2\},\max\{c_1,c_2\}+k^{(L)}-1],\\
    \operatorname{Shape}_x
    =
    (|r_2-r_1|+k^{(L)},\ |c_2-c_1|+k^{(L)}),
\end{gathered}$}
\end{equation*}
and $\operatorname{Mask}_x$ indicates the positions covered by the union
of the two LRF-Cs.
\end{myCorollary}

For a deeper-layer PSP, the recovered vector in layer $L'$ is the difference of two
convolution matrix rows, so the LRF-C union argument of
Proposition~\ref{prop:deeper_layer_rpcp} applies verbatim.

\begin{proposition}[Deeper-Layer PSP Supports Are Unions of LRF-Cs]
\label{prop:deeper_layer_psp}
For a PSP $x$ arising from a convolutional layer $L'>L$, the following
equation holds:
\begin{equation*}
    \psi^{(L)}(S_x)
    =
    \bigcup_{i\in\mathcal{M}_x}
    \bigl(
    [r_i,r_i+k^{(L)}-1]\times[c_i,c_i+k^{(L)}-1]
    \bigr),
\end{equation*}
where $\mathcal{M}_x$ indexes the LRF-Cs of layer $L$ associated with the
two rows that form the recovered vector of $x$. Consequently, a
deeper-layer PSP has the same union-of-LRF-Cs structure as a deeper-layer
RPCP, and $\operatorname{Mask}_x$ indicates the positions covered within
$\operatorname{Box}_x$.
\end{proposition}

\begin{proof}
The idea is the same as in Proposition~\ref{prop:deeper_layer_rpcp}, except that $S_x$
is the union over the LRF-Cs of the two rows forming the recovered vector
of $x$.
\end{proof}

\subsubsection{Identifying FCPs.}

The spatial characteristics of an FCP's recovered
vector are given by Proposition~\ref{prop:recovery_vector_sparsity_consistency}.
We formalize them as
Corollary~\ref{prop:target_layer_fcp}.

\begin{myCorollary}[Box, Shape, and Mask of FCPs]
\label{prop:target_layer_fcp}
For a target-layer or deeper-layer FCP $x$, the following hold:
\begin{equation*}
    S_x=\mathcal{I}^{(L)},
    \;\;
    \operatorname{Box}_x=\Lambda^{(L)},
    \;\;
    \operatorname{Shape}_x=(H^{(L)},W^{(L)}),
    \;\;
    \operatorname{Mask}_x=\mathbf{1}_{H^{(L)}\times W^{(L)}}.
\end{equation*}
\end{myCorollary}

\subsection{Method for Layer-Type Identification}
\label{subsec:layer_type_identification_algorithm}
We define the full-box point set as
    $\mathcal{P}_{\mathrm{full}}
    =
    \left\{
    x\in\mathcal{P}^{(L)}:
    \operatorname{Shape}_x=(H^{(L)},W^{(L)})
    \right\}$.
The remaining points form the set
    $\mathcal{P}_{\mathrm{conv}}
    =\mathcal{P}^{(L)}\setminus\mathcal{P}_{\mathrm{full}}$.
The results above determine the Box of every point in $\mathcal{P}^{(L)}$.
Combined with Lemma~\ref{lem:layer_point_composition}, they yield the following
decision rule.

Every point $x\in\mathcal{P}^{(L)}$ with
$\operatorname{Box}_x\neq\Lambda^{(L)}$ is an RPCP or a PSP, because an FCP
satisfies $\operatorname{Box}_x=\Lambda^{(L)}$. Therefore, if
$\mathcal{P}_{\mathrm{conv}}\neq\emptyset$, then
$\widehat{\tau}^{(L)}=\mathrm{Conv}$.

\iffalse
Algorithm~\ref{alg:layer_type_identification} gives the complete procedure for
identifying the type of layer $L$.

\begin{algorithm}[H]
\caption{Layer-type identification for layer $L$}
\label{alg:layer_type_identification}
\begin{algorithmic}[1]
   \REQUIRE target-layer input space $\Omega^{(L)}$; searched point set $\mathcal{P}^{(L)}$
   \ENSURE point sets $\mathcal{P}_{\mathrm{full}}$ and $\mathcal{P}_{\mathrm{conv}}$; layer type $\widehat{\tau}^{(L)}$
   \STATE Initialize $\mathcal{P}_{\mathrm{full}}\leftarrow\emptyset$
   \FOR{each point $x\in\mathcal{P}^{(L)}$}
       \STATE Compute $v_x$ over $\Omega^{(L)}$
       \STATE Form $S_x$ and compute $\operatorname{Shape}_x$
       \IF{$\operatorname{Shape}_x=(H^{(L)},W^{(L)})$}
           \STATE Add $x$ to $\mathcal{P}_{\mathrm{full}}$
       \ELSE
           \STATE Add $x$ to $\mathcal{P}_{\mathrm{conv}}$
       \ENDIF
   \ENDFOR
   \IF{$\mathcal{P}_{\mathrm{conv}}\neq\emptyset$}
       \STATE Set $\widehat{\tau}^{(L)}\leftarrow\mathrm{Conv}$
   \ELSE
       \STATE Set $\widehat{\tau}^{(L)}\leftarrow\mathrm{FC}$
   \ENDIF
   \STATE \textbf{return} $\mathcal{P}_{\mathrm{full}}$, $\mathcal{P}_{\mathrm{conv}}$, $\widehat{\tau}^{(L)}$
\end{algorithmic}
\end{algorithm}
\fi
% !TeX spellcheck = en_US
% !TEX root = ../main.tex

\section{Recovering the Convolution Kernel Size and Stride}
\label{sec:conv_parameter_identification}

In this section, our tasks are to recover the convolution kernel size
and the convolution stride, yielding $\widehat{k}$ and
$\widehat{s}_k$. In a convolutional
network, the size of each LRF-C is $k^{(L)}\times k^{(L)}$. Under
the convolution-matrix representation, the set
$\{j\in\mathcal{I}^{(L)}:(A_i^{(L)})_j\ne0\}$ is exactly the LRF-C
corresponding to the $i$-th row. By
Proposition~\ref{prop:recovery_vector_sparsity_consistency}, $S_x$
coincides with this LRF-C index set under Valid padding, whereas it does
not under Same padding. We therefore define an interior set to exclude
the potential influence of Same padding:
\begin{equation*}
    \widetilde{\mathcal{P}}_{\mathrm{conv}}^{\circ}
    =
    \left\{
    x\in\mathcal{P}_{\mathrm{conv}}:
    \operatorname{Box}_x\subseteq
    \{1,\ldots,H^{(L)}-2\}\times\{1,\ldots,W^{(L)}-2\}
    \right\}.
\end{equation*}
We propose
Lemma~\ref{lem:conv_kernel_size_recovery} to characterize the
relationship between $k^{(L)}$ and $\operatorname{Shape}_x$.

\begin{lemma}[Kernel Size Characterized by Shape and Mask]
\label{lem:conv_kernel_size_recovery}
For an interior target-layer point $x \in \widetilde{\mathcal{P}}_{\mathrm{conv}}^{\circ}$,
the following equation holds:
\begin{equation*}
    \label{eq:pointwise_kernel_estimate}
    k^{(L)}
    =
    \begin{cases}
        h_x = w_x,
        & x\text{ is an RPCP},\\
        b_{\partial}(\operatorname{Mask}_x) =\min\{h_x,w_x\},
        & x\text{ is a horizontal or vertical PSP},\\
        b_{\partial}(\operatorname{Mask}_x),
        & x\text{ is an off-axis PSP},
    \end{cases}
\end{equation*}
where $b_{\partial}(\operatorname{Mask}_x)$ denotes the length of the
shortest contiguous nonzero segment on the boundary of
$\operatorname{Mask}_x$.
\end{lemma}

\begin{proof}
It follows directly from Corollaries~\ref{prop:target_layer_rpcp} and~\ref{prop:target_layer_psp}.
\end{proof}

For an interior point $x$, $\min S_x$ is the position of $\bigl(K_0^{(L)}\bigr)_0$ in the
convolution matrix, so $\psi^{(L)}\bigl(\min S_x\bigr)$ is the start
position of an LRF-C. We propose
Lemma~\ref{lem:conv_window_start_recovery} to characterize the
relationship between $s_k^{(L)}$ and the starting positions.

\begin{lemma}[Stride Characterized by Effective Position Sets]
\label{lem:conv_window_start_recovery}
For any two interior target-layer points $x_1,x_2\in\widetilde{\mathcal{P}}_{\mathrm{conv}}^{\circ}$:
\begin{equation*}
    r_{x_1}-r_{x_2}\in s_k^{(L)}\mathbb{Z},
    %\qquad
    c_{x_1}-c_{x_2}\in s_k^{(L)}\mathbb{Z}.
\end{equation*}
\end{lemma}

\begin{proof}
Since convolution windows advance by $s_k^{(L)}$ in both directions,
the differences between the starting positions of any two windows, in both the row and the column direction, are integer multiples of $s_k^{(L)}$.
\end{proof}

Lemmas~\ref{lem:conv_kernel_size_recovery} and~\ref{lem:conv_window_start_recovery} transform Tasks~2 and~3 into
identifying the target-layer point.
In the following, we study how to identify the target-layer RPCPs and PSPs from $\widetilde{\mathcal{P}}_{\mathrm{conv}}^{\circ}$.

\subsection{Identifying the Target Points}
\label{subsec:target_point_identification}
We separately discuss the spatial and numerical properties of the recovered vectors
obtained in the target-layer input.

\subsubsection{Identifying Target-layer RPCP.}
The difference between the recovered vectors of RPCPs and PSPs is whether
$v_x$ equals one row of the convolution matrix or the difference of two
rows (Lemma~\ref{lem:recovery_vector_cnn}). 
Since the rows of the
convolution matrix are shifted copies of the
same kernel, a direct observation is that the support set $S_x$ of a
target-layer PSP is larger than that of a target-layer RPCP. 
Define the point set
    $\widetilde{\mathcal{P}}_{\operatorname{conv}}^{\circ}(k')
    =
    \{x\in\widetilde{\mathcal{P}}_{\mathrm{conv}}^{\circ}:
    \operatorname{Shape}_x=(k',k')\}$.
Since an LRF-C is a $k^{(L)}\times k^{(L)}$ square window, every point
$x\in\widetilde{\mathcal{P}}_{\operatorname{conv}}^{\circ}(k')$ is either a
target-layer RPCP, a target-layer off-axis PSP, or a deeper-layer point.

We remove distractions using the following filters:
\begin{itemize}
    \item 
    Retain candidates satisfying
    $\operatorname{Mask}_x=\mathbf{1}_{k'\times k'}$. This removes square
    off-axis PSPs, whose Masks are unions of two LRF-Cs.
    \item According to Proposition~\ref{prop:deeper_layer_rpcp}, for a
    deeper-layer point $x \in \widetilde{\mathcal{P}}_{\operatorname{conv}}^{\circ}(k')$,
    two deeper-layer points yield the same $v_x$ if and only if every
    intermediate layer exhibits the same ReLU activation
    arrangement~\cite{DBLP:conf/eurocrypt/LiuSELBP26}; this happens only
    when the pooling operation between $L$ and $L'$ is not max pooling.
    Therefore, if there are two or more RPCPs, we can retain the 
    correct target-layer RPCP set ${\mathcal{P}}_{\operatorname{RPCP}}^{\circ}$.
\end{itemize}

\subsubsection{Identifying Target-layer PSP.}
Define the remaining interior points as
    $\widetilde{\mathcal{P}}_{\mathrm{PSP}}^{\circ}
    =
    \widetilde{\mathcal{P}}_{\mathrm{conv}}^{\circ}
    \setminus
    \mathcal{P}_{\mathrm{RPCP}}^{\circ}$, 
which contains target-layer PSPs and deeper-layer points.
For $b\in\{\mathrm{top},\mathrm{bottom},\mathrm{left},\mathrm{right}\}$,
the four boundaries of $\operatorname{Box}_x=[u_{\min},u_{\max}]\times[v_{\min},v_{\max}]$
are defined as
\begin{equation*}
\begin{aligned}
\partial_{\mathrm{top}}\operatorname{Box}_x&=\{u_{\min}\}\times[v_{\min},v_{\max}],
&
\partial_{\mathrm{bottom}}\operatorname{Box}_x&=\{u_{\max}\}\times[v_{\min},v_{\max}],\\
\partial_{\mathrm{left}}\operatorname{Box}_x&=[u_{\min},u_{\max}]\times\{v_{\min}\},
&
\partial_{\mathrm{right}}\operatorname{Box}_x&=[u_{\min},u_{\max}]\times\{v_{\max}\}.
\end{aligned}
\end{equation*}
Denote the effective positions of $v_x$ on boundary $b$ by
\begin{equation*}
\label{eq:boundary_effective_positions}
\mathcal{B}_{x,b}
=
\bigl\{(\psi^{(L)})^{-1}(u,v)\in(\psi^{(L)})^{-1}\bigl(\partial_b\operatorname{Box}_x\bigr):\
(\operatorname{Mask}_x)_{u-u_{\min},\,v-v_{\min}}=1\bigr\}.
\end{equation*}
In the following, $K_{:,j}^{(L)}$ (resp. $K_{i,:}^{(L)}$) denotes the
$j$-th column (resp. $i$-th row) of the kernel $K^{(L)}$.
We extend Proposition~\ref{prop:recovery_vector_numerical_consistency} to
Proposition~\ref{prop:boundary_kernel_equality} to filter deeper-layer points.

\begin{proposition}[Boundary Equality with the Kernel]
\label{prop:boundary_kernel_equality}
For any target-layer PSP $x \in \widetilde{\mathcal{P}}_{\mathrm{PSP}}^{\circ}$,
the following hold:
\begin{itemize}
    \item horizontal PSP: $v_x\bigl[\mathcal{B}_{x,\mathrm{left}}\bigr]=\pm g_x\,K^{(L)}_{:,0}$
    and $v_x\bigl[\mathcal{B}_{x,\mathrm{right}}\bigr]=\pm g_x\,K^{(L)}_{:,k^{(L)}-1}$;
    \item vertical PSP: $v_x\bigl[\mathcal{B}_{x,\mathrm{top}}\bigr]=\pm g_x\,K^{(L)}_{0,:}$
    and $v_x\bigl[\mathcal{B}_{x,\mathrm{bottom}}\bigr]=\pm g_x\,K^{(L)}_{k^{(L)}-1,:}$;
    \item off-axis PSP: $v_x\bigl[\mathcal{B}_{x,\mathrm{top}}\bigr]=\pm g_x\,K^{(L)}_{0,:}$,
    $v_x\bigl[\mathcal{B}_{x,\mathrm{left}}\bigr]=\pm g_x\,K^{(L)}_{:,0}$,
    $v_x\bigl[\mathcal{B}_{x,\mathrm{bottom}}\bigr]=\pm g_x\,K^{(L)}_{k^{(L)}-1,:}$,
    and $v_x\bigl[\mathcal{B}_{x,\mathrm{right}}\bigr]=\pm g_x\,K^{(L)}_{:,k^{(L)}-1}$.
\end{itemize}
\end{proposition}

\begin{proof}
It follows from the planar relationship between the two LRF-Cs of each PSP type.
\end{proof}

Proposition~\ref{prop:boundary_kernel_equality} expresses the boundary entries
of $v_x$ in terms of the kernel, but each entry is scaled by the unknown factor
$\pm g_x$, which is specific to $x$. Taking the ratio of two such entries
eliminates $g_x$ and leaves a quantity determined by the kernel alone. For any
two PSPs $x_1,x_2\in\widetilde{\mathcal{P}}_{\mathrm{PSP}}^{\circ}$, the ratios
of the first to the last nonzero entries of $v_{x_1}$ and $v_{x_2}$ are
therefore equal, and both equal the ratio of the top-left to the bottom-right
entry of the kernel:
\begin{equation}
\label{eq:equal_endpoint_ratios}
\frac{v_{x_1}\bigl[\min S_{x_1}\bigr]}{v_{x_1}\bigl[\max S_{x_1}\bigr]}
=
\frac{v_{x_2}\bigl[\min S_{x_2}\bigr]}{v_{x_2}\bigl[\max S_{x_2}\bigr]}
=
\frac{\bigl(K_0^{(L)}\bigr)_0}{\bigl(K_{k^{(L)}-1}^{(L)}\bigr)_{k^{(L)}-1}}.
\end{equation}
The endpoints are only the simplest instance of this consistency. By
Proposition~\ref{prop:boundary_kernel_equality}, the whole non-overlapping part
of a PSP, not just its two extremes, reproduces the kernel up to the factor
$\pm g_x$, so the ratio of any two such entries is likewise $x$-independent;
the overlap of the two LRF-Cs is excluded, since there $v_x$ is a difference of
shifted kernel copies. In practice the two forms are used progressively.

Since $g_x$ cancels, the common ratio in
Eq.~\eqref{eq:equal_endpoint_ratios} can be evaluated from a recovered vector
alone, without prior knowledge of the kernel. We use it as an $x$-independent
invariant to group PSPs sharing the same kernel and thereby identify the
target-layer PSP set $\mathcal{P}_{\mathrm{PSP}}^{\circ}$ in two cases:
\begin{itemize}
    \item \textit{$\mathcal{P}_{\operatorname{RPCP}}^{\circ}\neq\emptyset$.}
    The kernel size $k^{(L)}$ and the kernel parameters are directly recovered
    from the RPCPs. The target-layer PSPs are then identified from both
    the spatial and the numerical aspects (Eqs.~\eqref{eq:pointwise_kernel_estimate} and \eqref{eq:equal_endpoint_ratios}).

    \item \textit{$\mathcal{P}_{\operatorname{RPCP}}^{\circ}=\emptyset$.}
    We apply the same numerical clustering as in the RPCP identification:
    cluster the candidates by Eq.~\eqref{eq:equal_endpoint_ratios}.
\end{itemize}
After target-layer point identification, the retained interior RPCPs and PSPs
form
    $\mathcal{P}_{\mathrm{conv}}^{\circ}
    =
    \mathcal{P}_{\mathrm{RPCP}}^{\circ}
    \cup
    \mathcal{P}_{\mathrm{PSP}}^{\circ}$.

\subsection{Method for Convolution Kernel Size and Stride Recovery}
\label{subsec:conv_stride_identification}

When $\mathcal{P}_{\operatorname{RPCP}}^{\circ}\neq\emptyset$, the identified
kernel size $\widehat{k}$ is directly given by
Lemma~\ref{lem:conv_kernel_size_recovery}, and the target-layer PSP set
$\mathcal{P}_{\mathrm{PSP}}^{\circ}$ is filtered by
Lemma~\ref{lem:conv_kernel_size_recovery} together with
numerical consistency in non-overlapping areas. When
$\mathcal{P}_{\operatorname{RPCP}}^{\circ}=\emptyset$, we vote among the PSPs
for the kernel size: the vote count for $k'$ is
    $V(k')
    =
    \sum_{x\in\widetilde{\mathcal{P}}_{\mathrm{PSP}}^{\circ}}
    \left[b_{\partial}(\operatorname{Mask}_x)=k'\right]$.
The identified kernel size is $\widehat{k}=\arg\max_{k'\in\mathbb{N}}V(k')$,
and the points supporting this result are retained as
$\mathcal{P}_{\mathrm{PSP}}^{\circ}$.

We collect the starting positions of the effective position sets to form
    $\mathcal{S}_{\mathrm{conv}}
    =
    \{(r_x,c_x):
    x\in\mathcal{P}_{\mathrm{conv}}^{\circ}\}$.
We then collect its nonzero row and column differences to form
\begin{equation*}
\begin{aligned}
    \Delta_{\mathrm{conv}}
    ={}&
    \left\{|r-r'|:
    (r,c),(r',c')\in\mathcal{S}_{\mathrm{conv}},\ r\neq r'\right\}
    \\
    &\cup
    \left\{|c-c'|:
    (r,c),(r',c')\in\mathcal{S}_{\mathrm{conv}},\ c\neq c'\right\}.
\end{aligned}
\end{equation*}

By Lemma~\ref{lem:conv_window_start_recovery}, the convolution stride is recovered as
    $\widehat{s}_k
    =
    \min\Delta_{\mathrm{conv}}$.
% !TeX spellcheck = en_US
% !TEX root = ../main.tex
\section{Recovering the Padding Mode and the Output Channel Number}
\label{sec:padding_output_channel_recovery}

In this section, our tasks are to recover the padding mode (Task~4) and the
output-channel number (Task~5), yielding
$\widehat{\pi}\in\{\mathrm{Valid},\mathrm{Same}\}$
and $\widehat{C}_{\mathrm{out}}$. 

Under Same padding, the input is padded with zeros, and part of the
LRF-C of a point may lie in the zero-padding region. The kernel entries
at the padded positions are not exposed in the true input coordinates,
so the recovered vector is truncated on the corresponding side.
We now state Lemma~\ref{lem:padding_boundary}, which gives a
sufficient condition for identifying Same padding.

\begin{lemma}[Truncated Points Imply Same Padding]
\label{lem:padding_boundary}
A target-layer truncated point is a point whose LRF-C intersects the
zero-padding region. If such a point exists, then
$\pi^{(L)}=\mathrm{Same}$.
\end{lemma}

\begin{proof}
Under Valid padding, no zero-padding region exists, so no LRF-C
intersects it.
\end{proof}

Lemma~\ref{lem:padding_boundary} transforms Task~4 into identifying
target-layer truncated points among the boundary candidates
$\widetilde{\mathcal{P}}_{\partial}$, while Task~5 is accomplished by clustering
the retained target-layer points.

\subsection{Identifying the Truncated Point}
\label{subsec:padding_identification}

We define the boundary candidates as
    $\widetilde{\mathcal{P}}_{\partial}
    =
    \mathcal{P}_{\mathrm{conv}}
    \setminus
    \widetilde{\mathcal{P}}_{\mathrm{conv}}^{\circ}$.
Every candidate is a target-layer RPCP or PSP whose Box touches the
boundary, a target-layer truncated point, or a deeper-layer point.

Let
$\widetilde{\mathcal{P}}_{\partial,\mathrm{left}}
=\{x\in\widetilde{\mathcal{P}}_{\partial}:
\operatorname{Box}_x\text{ touches the left boundary of }\Lambda^{(L)}\}$
denote the left-boundary candidates.
Since the four boundaries are symmetric, we take the left
boundary as an example to describe the method.

A target-layer truncated point may be obtained by truncating an RPCP or
PSP at the input boundary. Based on
Proposition~\ref{prop:boundary_kernel_equality}, we derive Proposition~\ref{prop:truncated_boundary_equality}.

\begin{proposition}[Boundary Equality of Truncated Points]
\label{prop:truncated_boundary_equality}
For a target-layer truncated point
$x\in\widetilde{\mathcal{P}}_{\partial,\mathrm{left}}$ obtained by truncating a
target-layer RPCP, horizontal PSP, or off-axis PSP:
    $v_x\bigl[\mathcal{B}_{x,\mathrm{right}}\bigr]
    =\pm g_x\,K^{(L)}_{:,k^{(L)}-1}$ and $v_x\bigl[\mathcal{B}_{x,\mathrm{left}}\bigr] \ne \pm g_x\,K^{(L)}_{:,0}$.
\end{proposition}

\begin{proof}
By Lemma~\ref{lem:recovery_vector_cnn}, $v_x$ is, up to the sign
absorbed in $g_x$, the restriction of the convolution-matrix row (resp.
row difference) to its LRF-C. Left truncation removes the leftmost
kernel columns, which fall into the zero-padding region, while the
right boundary of $\operatorname{Box}_x$ lies strictly inside the input
and its effective positions carry the complete last kernel column;
hence $v_x\bigl[\mathcal{B}_{x,\mathrm{right}}\bigr]
=\pm g_x\,K^{(L)}_{:,k^{(L)}-1}$. The effective positions on the left
boundary correspond to a truncated view rather than the first kernel
column, so for a generic kernel they do not coincide with
$\pm g_x\,K^{(L)}_{:,0}$.
\end{proof}

To filter the target-layer points, we select an RPCP from
$\mathcal{P}_{\mathrm{RPCP}}^{\circ}$ as the numerical template
point; if $\mathcal{P}_{\mathrm{RPCP}}^{\circ}=\emptyset$, we select a
horizontal PSP or off-axis PSP instead. We adopt the two numerical
verification conditions in
Proposition~\ref{prop:truncated_boundary_equality}, and denote the matched
point set by $\mathcal{P}_{\partial,\mathrm{left}}$. If
$\mathcal{P}_{\partial,\mathrm{left}}\neq\emptyset$, then a truncated
point exists among the left-boundary candidates.

By symmetry, we obtain the matched sets
$\mathcal{P}_{\partial,\mathrm{right}}$,
$\mathcal{P}_{\partial,\mathrm{top}}$, and
$\mathcal{P}_{\partial,\mathrm{bottom}}$ for the other three boundaries.
We denote the set of truncated points by
    $\mathcal{P}_{\partial}
    =
    \mathcal{P}_{\partial,\mathrm{left}}
    \cup\mathcal{P}_{\partial,\mathrm{right}}
    \cup\mathcal{P}_{\partial,\mathrm{top}}
    \cup\mathcal{P}_{\partial,\mathrm{bottom}}$.

\subsection{Method for Recovering the Padding Mode and the Output Channel Number}
\label{subsec:conv_parameter_recovery_algorithm}

According to Lemma~\ref{lem:padding_boundary}, if
$\mathcal{P}_{\partial}\neq\emptyset$, then a truncated point exists and
we recover $\pi^{(L)}=\mathrm{Same}$; otherwise,
$\pi^{(L)}=\mathrm{Valid}$.

We cluster the retained target-layer points into three groups. The points
in $\mathcal{P}_{\mathrm{RPCP}}^{\circ}$ are clustered by
Proposition~\ref{prop:recovery_vector_numerical_consistency}. The points
in $\mathcal{P}_{\mathrm{PSP}}^{\circ}$ are first classified into the
three PSP types according to the recovered kernel size $\widehat{k}$ and
stride $\widehat{s}_k$, and then clustered by
Proposition~\ref{prop:boundary_kernel_equality}. The truncated-point
sets composing $\mathcal{P}_{\partial}$ are clustered separately:
$\mathcal{P}_{\partial,\mathrm{left}}$,
$\mathcal{P}_{\partial,\mathrm{right}}$,
$\mathcal{P}_{\partial,\mathrm{top}}$, and
$\mathcal{P}_{\partial,\mathrm{bottom}}$; e.g.,
$\mathcal{P}_{\partial,\mathrm{left}}$ is clustered by the equality
condition
$v_x\bigl[\mathcal{B}_{x,\mathrm{right}}\bigr]
=\pm g_x\,K^{(L)}_{:,k^{(L)}-1}$
in Proposition~\ref{prop:truncated_boundary_equality}. Finally, all
resulting clusters are merged according to the kernel values exposed on
their non-overlapping regions, and the number of merged classes is
the recovered output-channel number $\widehat{C}_{\mathrm{out}}$.

% !TeX spellcheck = en_US
% !TEX root = ../main.tex

\section{Recovering the Pooling Window Size and Stride}
\label{subsec:pooling_structure_identification}

In this section, we recover the pooling type
$\widehat{\rho}\in\{\mathrm{NoPool},\mathrm{AvgPool},\mathrm{MaxPool}\}$.
If the type is $\mathrm{AvgPool}$ or $\mathrm{MaxPool}$, we further
recover the window size $\widehat{p}$ and the stride $\widehat{s}_p$.

PSPs are produced by the additional nonlinear operation introduced by
max pooling. Therefore, when
$\mathcal{P}_{\mathrm{PSP}}^{\circ}\neq\emptyset$, or when a PSP exists
in the boundary set $\mathcal{P}_{\partial}$, we determine
$\widehat{\rho}=\mathrm{MaxPool}$.
Otherwise, the two remaining cases, NoPool and
AvgPool, are both linear operations, so the pooling matrix formed under
the fully-connected view is fixed: $P^{(L)}=I$ for NoPool, while the
AvgPool matrix is a constant block matrix whose rows take
the constant value $1/p^2$ on one $p\times p$ LRF-P and are zero
elsewhere. Using this feature, we identify AvgPool and recover the window
size $\widehat{p}$ and the stride $\widehat{s}_p$.

\subsection{Recovering the Max-Pooling Window Size and Stride}
\label{subsec:max_pooling_parameter_identification}

Adjacent values within a pooling receptive field are computed by
adjacent convolution operations, whose window starts are separated by
$s_k^{(L)}$. For a PSP $x$, let
$(r'_x,c'_x):=\psi^{(L)}\bigl(\max S_x\bigr)-(k^{(L)}-1,k^{(L)}-1)$
denote the start of its second LRF-C. We propose the following result.

\begin{proposition}[PSP Window-Start Offsets Are Stride Multiples]
\label{prop:psp_window_start_offset}
For any PSP $x \in \mathcal{P}_{\mathrm{PSP}}^{\circ}$, the offsets
between its two LRF-C starts satisfy:
\begin{equation}
    \label{eq:psp_window_start_offset}
    \begin{aligned}
    \bigl|\,r_x-r'_x\,\bigr| &= a\,s_k^{(L)},\\
    \bigl|\,c_x-c'_x\,\bigr| &= a\,s_k^{(L)},
    \end{aligned}
\end{equation}
where $a\in\{0,\ldots,p-1\}$ in each spatial dimension.
\end{proposition}

\begin{proof}
The coordinates $(r_x,c_x)$
are the start of one LRF-C of the PSP, and
$(r'_x,c'_x)$
are the start of the other LRF-C. The differences between the two LRF-Cs
in each spatial dimension are integer multiples of the stride $s_k^{(L)}$.
Moreover, the outputs of the two LRF-Cs lie within the same pooling
window, so the offset between the two LRF-C starts is at most $(p-1)$
stride units in each dimension.
Therefore, Eq.~\eqref{eq:psp_window_start_offset} holds.
\end{proof}

For the PSPs $x\in\mathcal{P}_{\mathrm{PSP}}^{\circ}$, evaluating
the left-hand side of Eq.~\eqref{eq:psp_window_start_offset} for each PSP $x$ forms the value set $\mathcal{D}$. By
Proposition~\ref{prop:psp_window_start_offset}, we transform the
max-pooling window-size identification into selecting the maximum value
of $\mathcal{D}$ and computing $\widehat{p}=(\max\mathcal{D})/s_k^{(L)}+1$.

We define the point sets whose offsets attain the maximum value of $\mathcal{D}$:
\begin{equation*}
    \mathcal{P}_{\max}^{r}
    =
    \Bigl\{
    x\in\mathcal{P}_{\mathrm{PSP}}^{\circ}:
    \bigl|\,r_x-r'_x\,\bigr|
    =(\widehat{p}-1)\,s_k^{(L)}
    \Bigr\},
\end{equation*}
\begin{equation*}
    \mathcal{P}_{\max}^{c}
    =
    \Bigl\{
    x\in\mathcal{P}_{\mathrm{PSP}}^{\circ}:
    \bigl|\,c_x-c'_x\,\bigr|
    =(\widehat{p}-1)\,s_k^{(L)}
    \Bigr\}.
\end{equation*}
We then propose Proposition~\ref{prop:max_pooling_stride} to characterize the
relationship between the pooling-window anchors and the pooling stride.

\begin{proposition}[Max-Pooling Anchors Lie on the Stride Grid]
\label{prop:max_pooling_stride}
For any two points $x_1,x_2$ of $\mathcal{P}_{\max}^{r}$,
    $r_{x_1}-r_{x_2}
    \in s_p s_k^{(L)}\mathbb{Z}$.
Symmetrically, for any two points
$x_1,x_2$ of $\mathcal{P}_{\max}^{c}$,
    $c_{x_1}-c_{x_2}
    \in s_p s_k^{(L)}\mathbb{Z}$.
\end{proposition}

\begin{proof}
Pooling windows advance by $s_p$ response positions. After mapping through the
target convolution, their anchors therefore advance by $s_ps_k^{(L)}$ input
coordinates.
\end{proof}

We collect the coordinates of the retained maximum-offset points:
    $\mathcal{R}_{\max}
    =
    \{r_x:x\in\mathcal{P}_{\max}^{r}\}$
    and
    $\mathcal{V}_{\max}
    =
    \{c_x:x\in\mathcal{P}_{\max}^{c}\}$.
These coordinates are measured in the target-layer input space. We therefore
compute their nonzero differences in units of the recovered convolution stride:
\begin{equation*}
    \Delta
    =
    \left\{
    \frac{|r-r'|}{\widehat{s}_k}:
    r,r'\in\mathcal{R}_{\max},\ r\neq r'
    \right\}
    \cup
    \left\{
    \frac{|c-c'|}{\widehat{s}_k}:
    c,c'\in\mathcal{V}_{\max},\ c\neq c'
    \right\}.
\end{equation*}
By Proposition~\ref{prop:max_pooling_stride}, every element of
$\Delta$ is a positive integer multiple of $s_p$. If the retained PSPs
include two adjacent pooling-window anchors, 
we recover the max-pooling stride
as $\widehat{s}_p = \min\Delta$.

\subsection{Recovering the Average-Pooling Window Size and Stride}
\label{subsec:average_pooling_parameter_identification}

After max pooling has been excluded,
we recover the pooling type $\widehat{\rho}$ and the pooling parameters.
We adopt a skip-and-determine strategy: skip the pooling operation in layer $L$ and recover the parameters
through the recovered vectors in layer $L+1$.
Let $v_x^{(L+1)}$ denote the recovered
vector in the actual layer-$(L+1)$ input coordinates and let
$\widetilde{v}_x^{(L+1)}$ denote the vector recovered after skipping the pooling operation.
For a layer-$(L+1)$ RPCP $x$, its recovered vector satisfies
\begin{equation*}
v_x^{(L+1)}
= \widetilde{v}_x^{(L+1)}P^{(L)}
= g_xA_{m,i}^{(L+1)}[\mathcal{K}_i^{(L+1)}],
\end{equation*}
and for a layer-$(L+1)$ PSP $x$,
\begin{equation*}
v_x^{(L+1)}
= \widetilde{v}_x^{(L+1)}P^{(L)}
= g_x\bigl(A_{m,i}^{(L+1)}-A_{m,j}^{(L+1)}\bigr)[\mathcal{K}_i^{(L+1)}].
\end{equation*}

If $\rho^{(L)}=\mathrm{NoPool}$, then $P^{(L)}=I$ and $\widetilde{v}_x^{(L+1)} = v_x^{(L+1)}$.
If $\rho^{(L)}=\mathrm{AvgPool}$, we describe the structure of the restored vector in Proposition~\ref{prop:avg_pooling_equal_positions}.

Suppose $\rho^{(L)}=\mathrm{AvgPool}$ with window size $p$ and stride
$s_p$, and let $w_\rho^{(L)}$ denote the pooling output width. For a
pooled coordinate $i\in\{0,\ldots,d^{(L+1)}-1\}$, the flattened LRF-P
of $i$ (the pooling analogue of $\mathcal{K}_i^{(L)}$) is
\begin{equation*}
    \mathcal{K}_i^{P}
    =
    \left\{
    \bigl(\bigl\lfloor i/w_\rho^{(L)}\bigr\rfloor s_p+j\bigr)w_o^{(L)}
    +\bigl(i\bmod w_\rho^{(L)}\bigr)s_p+m:\
    0\le j<p,\ 0\le m<p
    \right\},
\end{equation*}
i.e., the set of pre-pooling positions averaged into the pooled
coordinate $i$; the $i$-th row of $P^{(L)}$ equals $\tfrac{1}{p^2}$ on
$\mathcal{K}_i^{P}$ and $0$ elsewhere. We present the single-channel
case; channels pool independently.

\begin{proposition}[Average-Pooling Structure of the Restored Vector]
\label{prop:avg_pooling_equal_positions}
Let $x$ be a layer-$(L+1)$ point with recovered vector
$v_x^{(L+1)}=\widetilde{v}_x^{(L+1)}P^{(L)}$, and let $S_x^{(L+1)}$
denote the effective position set of $v_x^{(L+1)}$. Define the
restored vector $\widehat{v}_x^{(L+1)}$ by
\begin{equation}
    \label{eq:avgpool_entry_formula}
    \bigl(\widehat{v}_x^{(L+1)}\bigr)_{r}
    =
    p^{2}\sum_{i':\ r\in\mathcal{K}_{i'}^{P}}
    \bigl(v_x^{(L+1)}\bigr)_{i'}
    \qquad\text{for every pre-pooling position } r.
\end{equation}
Then the following hold.
\begin{enumerate}
    \item $\bigl(\widehat{v}_x^{(L+1)}\bigr)_r=0$ outside
    $\bigcup_{i'\in S_x^{(L+1)}}\mathcal{K}_{i'}^{P}$;
    \item if $s_p\ge p$, then $\widehat{v}_x^{(L+1)}$ is constant on
    each $\mathcal{K}_i^{P}$ with $i\in S_x^{(L+1)}$, with value
    $p^{2}\bigl(v_x^{(L+1)}\bigr)_i$, and consecutive segments are
    separated by $g=s_p-p$ zeros;
    \item if $s_p<p$, then $\widehat{v}_x^{(L+1)}$ equals
    $p^{2}\bigl(v_x^{(L+1)}\bigr)_i$ on the exclusive part (of length
    $\ell=s_p$) of each LRF-P, whereas each overlapped position carries
    the sum of the adjacent restored values.
\end{enumerate}
\end{proposition}

\begin{proof}
Since the $i$-th row of $P^{(L)}$ equals $\tfrac{1}{p^{2}}$ exactly on
$\mathcal{K}_i^{P}$ and zero elsewhere, the $i$-th coefficient of
$v_x^{(L+1)}=\widetilde{v}_x^{(L+1)}P^{(L)}$ is
\[
\bigl(v_x^{(L+1)}\bigr)_i
=
\widetilde{v}_x^{(L+1)}\bigl(P_i^{(L)}\bigr)^{\top}
=
\tfrac{1}{p^{2}}
\sum_{r\in\mathcal{K}_i^{P}}
\bigl(\widetilde{v}_x^{(L+1)}\bigr)_{r}.
\]
By Eq.~\eqref{eq:recovered_signature_vector},
$\bigl(v_x^{(L+1)}\bigr)_i\neq0$ holds only for
$i\in S_x^{(L+1)}$, so every term in
Eq.~\eqref{eq:avgpool_entry_formula} vanishes unless
$i'\in S_x^{(L+1)}$; this proves case 1.
For case 2 and case 3, we count how many flattened LRF-Ps
$\mathcal{K}_{i'}^{P}$ with $i'\in S_x^{(L+1)}$ cover each
pre-pooling position:
\begin{itemize}
    \item \textit{Case $s_p\ge p$ (disjoint windows).} Each position
    $r\in\mathcal{K}_i^{P}$ with $i\in S_x^{(L+1)}$ lies in no other
    such LRF-P, so the sum in Eq.~\eqref{eq:avgpool_entry_formula}
    reduces to the single term $p^{2}\bigl(v_x^{(L+1)}\bigr)_i$;
    consecutive segments are separated by $g=s_p-p$ zero positions.
    This proves (2).
    \item \textit{Case $s_p<p$ (overlapping windows).} The exclusive
    part of each $\mathcal{K}_i^{P}$, of length $\ell=s_p$, is covered
    by no other such LRF-P and therefore carries the constant value
    $p^{2}\bigl(v_x^{(L+1)}\bigr)_i$, whereas each overlapped position
    receives the sum of the corresponding restored values. This proves
    (3).
\end{itemize}
\end{proof}

Based on Proposition~\ref{prop:avg_pooling_equal_positions}, we obtain
Corollary~\ref{cor:avg_pool_parameter_estimate}, where $\ell$, $g$,
and $o$ denote the constant-segment length, the zero-gap length, and
the overlap length read from the entry sequence of
$\widehat{v}_x^{(L+1)}$.

\begin{myCorollary}[Average-Pooling Parameter Estimate]
\label{cor:avg_pool_parameter_estimate}
For a layer-$(L+1)$ point $x$ exhibiting this pattern, the following hold:
\begin{equation*}
    \label{eq:avg_pool_parameter_estimate}
    ( p, s_p)
    =
    \begin{cases}
        (\ell,\ell), & s_p=p,\\
        (\ell,\ell+g), & s_p>p,\\
        (\ell+o,\ell), & s_p<p.
    \end{cases}
\end{equation*}
\end{myCorollary}

\begin{proof}
By Proposition~\ref{prop:avg_pooling_equal_positions}, $\ell=p$ when
$s_p\ge p$, with gap $g=s_p-p$, and $\ell=s_p$ with overlap $o=p-s_p$
when $s_p<p$; substituting these relations into the three cases gives
the displayed pairs.
\end{proof}

Therefore, no pooling and average pooling are distinguished by testing
whether the entries of the restored vector $\widehat{v}_x^{(L+1)}$ at
adjacent positions
are equal: under average pooling, the entry sequence of
$\widehat{v}_x^{(L+1)}$ exhibits repeated constant blocks, i.e., the
entries within each flattened LRF-P are equal
(Proposition~\ref{prop:avg_pooling_equal_positions}); under no
pooling, $P^{(L)}=I$ and $\widehat{v}_x^{(L+1)}=v_x^{(L+1)}$ retains
the sparse convolution structure, so no such repeated-block pattern
appears. Once average pooling is identified, reading the segment length
$\ell$, the zero-gap length $g$, and the overlap length $o$ from the
entry sequence and substituting them into
Corollary~\ref{cor:avg_pool_parameter_estimate} completes the
parameter recovery of average pooling.

% !TeX spellcheck = en_US
% !TEX root = ../main.tex

\section{Experiments}
\label{sec:experiments}

This section evaluates the effectiveness of the proposed architecture recovery
attack and is organized as follows.
Section~\ref{subsec:existing_models} summarizes the overall experimental
results, and Section~\ref{subsec:target_model} presents visualizations of the
architecture recovery process, task by task, on one representative model.
All the experiments are conducted on a server equipped with two Intel Xeon Gold
5318Y CPUs at 2.10 GHz, 503 GB of memory.

\subsection{Summary of Experimental Results}
\label{subsec:existing_models}

We combine our architecture recovery attack with the existing parameter
recovery attacks for CNNs, using the attack of~\cite{cryptoeprint:2026/241} in
the raw-output setting, the attack of~\cite{cryptoeprint:2026/1164} in the
hard-label setting, and both the layer-by-layer and the end-to-end pipeline.
The tested models are a $(2+1)$ CNN and a $(2+2)$ CNN, the latter a LeNet-5 variant, both adopted by
these attacks~\cite{cryptoeprint:2026/241,cryptoeprint:2026/1164}. All models
use two convolutional layers with $5\times5$ kernels, stride $1$, Valid padding
and $2\times2$ max pooling, and they differ only in the output-channel numbers.

For every model under each tested setting, the architecture of each layer is
recovered exactly, and the subsequent parameter recovery achieves the same
accuracy as the original attacks.

A further concern is the complexity introduced by architecture recovery.
Its only cost is the time to analyze the recovered vectors that the
point search and equation solving have already produced; since the two
phases share this step, the queries consumed by architecture recovery
are a subset of those consumed by parameter recovery. For the $(2+1)$ CNN in the
raw-output setting, the complete parameter recovery requires $2^{20.48}$
queries and $2^{7.31}$ seconds, whereas the architecture-recovery phase
accounts for $2^{17.71}$ queries and $2^{4.08}$ seconds.

These experimental results fully prove that our attack works well in recovering the 
whole network architecture without any prior architectural knowledge, and it is 
feedback-agnostic. The complete results of all tested models are available in our repository.

\subsection{Visualizing the Architecture Recovery Process}
\label{subsec:target_model}

To demonstrate that our method handles diverse architectural
hyperparameters, we design a CNN model as the target.
Its convolutional layers deliberately employ different kernel sizes,
strides, padding modes, output-channel numbers, and pooling operations.
The complete architecture is shown in
Table~\ref{tab:calibrated_cnn_architecture}.
\begin{table}[!htb]
\centering
\caption{Ground-truth architecture of the $(3+1)$
CNN.}
\label{tab:calibrated_cnn_architecture}
\footnotesize
\setlength{\tabcolsep}{4pt}
\begin{adjustbox}{max width=\textwidth}
\begin{tabular}{clcccccc}
\toprule
Layer & Type & Output shape & $k$ & $s_k$ & $\pi$ & $C_{\mathrm{out}}$
& $\rho\,(p,s_p)$ \\
\midrule
1 & Conv & $1\times32\times32$ & 3 & 1 & Same  & 1 & NoPool \\
2 & Conv & $2\times16\times16$ & 5 & 1 & Same  & 2 & AvgPool $(2,2)$ \\
3 & Conv & $3\times3\times3$   & 3 & 2 & Valid & 3 & MaxPool $(3,2)$ \\
4 & FC & $18$ & -- & -- & -- & -- & -- \\
5 & Output & $10$ & -- & -- & -- & -- & -- \\
\bottomrule
\end{tabular}
\end{adjustbox}
\end{table}

This model was trained on the MNIST training set resized to $32\times32$,
achieving a training accuracy of 94.03\% and a test accuracy of 89.30\%.

We recover the architecture of this model layer by layer in the raw-output
setting. The point collection and the equation solving are taken from the
attack of~\cite{cryptoeprint:2026/241}, and our architecture recovery is
inserted between them. For uniformity, we fix the number of critical points
searched for each layer at $700$. For each task we
visualize the statistics collected during the recovery, showing how these
statistics determine the corresponding architecture
hyperparameter.

\paragraph{\normalfont\bfseries Task 1: Layer-Type Identification.} Figure~\ref{fig:task1_layer_type_counts} reports
$|\mathcal{P}_{\mathrm{full}}|$ and $|\mathcal{P}_{\mathrm{conv}}|$ for three layers.
Since $\mathcal{P}_{\mathrm{conv}}\neq\emptyset$ for every layer, all three
layers are identified as convolutional by Lemma~\ref{lem:layer_point_composition}. 

\begin{figure}[!htb]
\centering
\begin{minipage}[t]{0.32\textwidth}
    \centering
    \includegraphics[width=\linewidth]{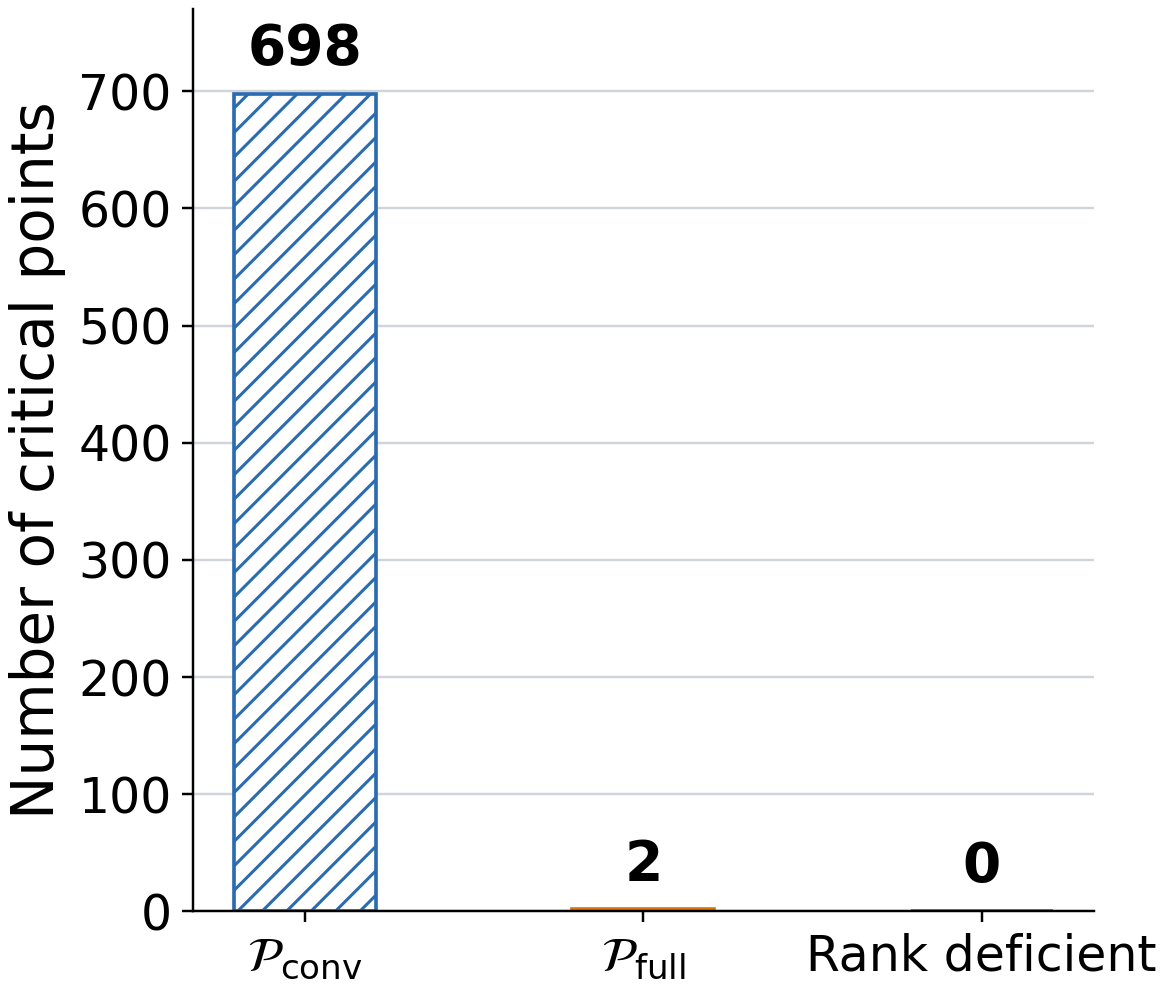}
    \small (a) Layer 1
\end{minipage}
\hfill
\begin{minipage}[t]{0.32\textwidth}
    \centering
    \includegraphics[width=\linewidth]{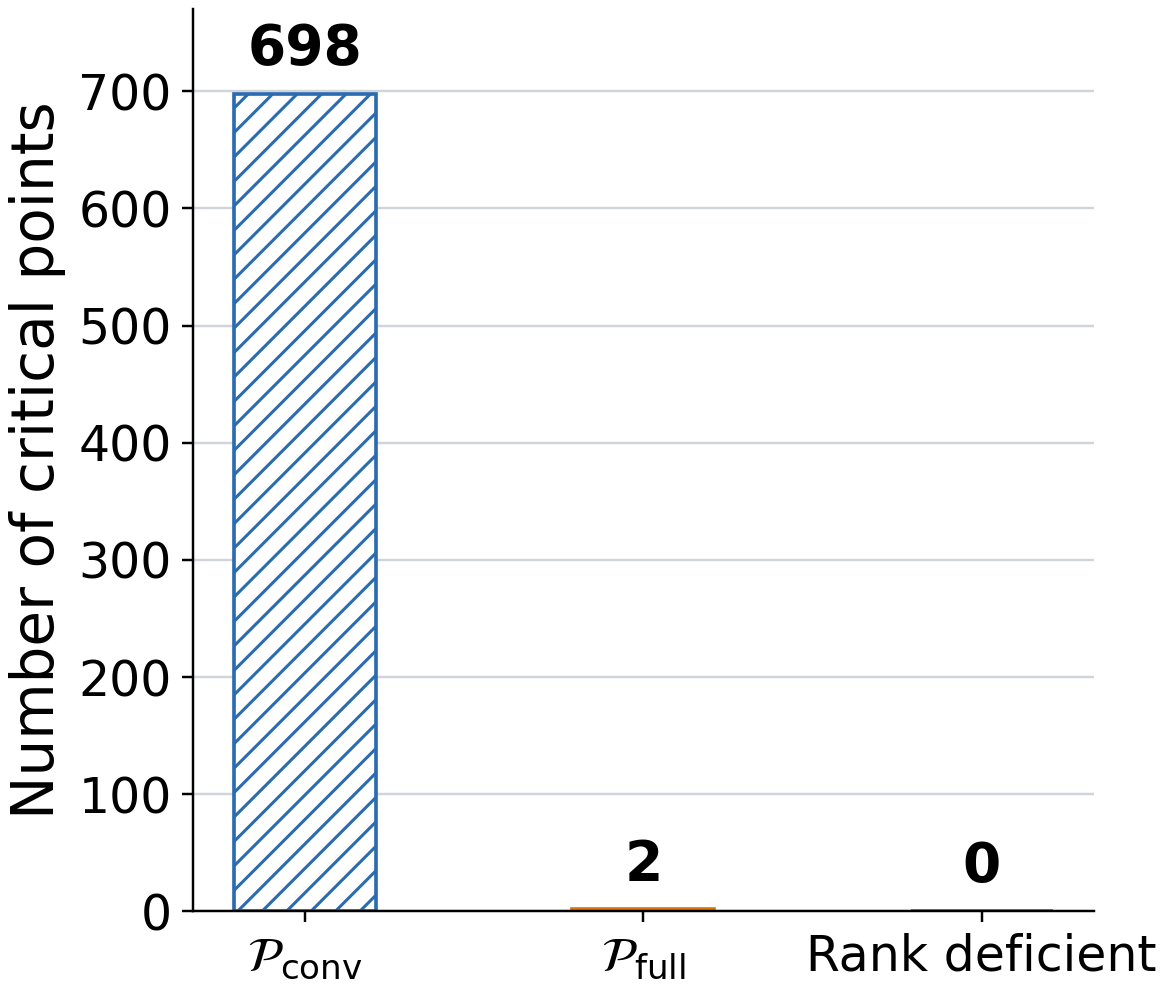}
    \small (b) Layer 2
\end{minipage}
\hfill
\begin{minipage}[t]{0.32\textwidth}
    \centering
    \includegraphics[width=\linewidth]{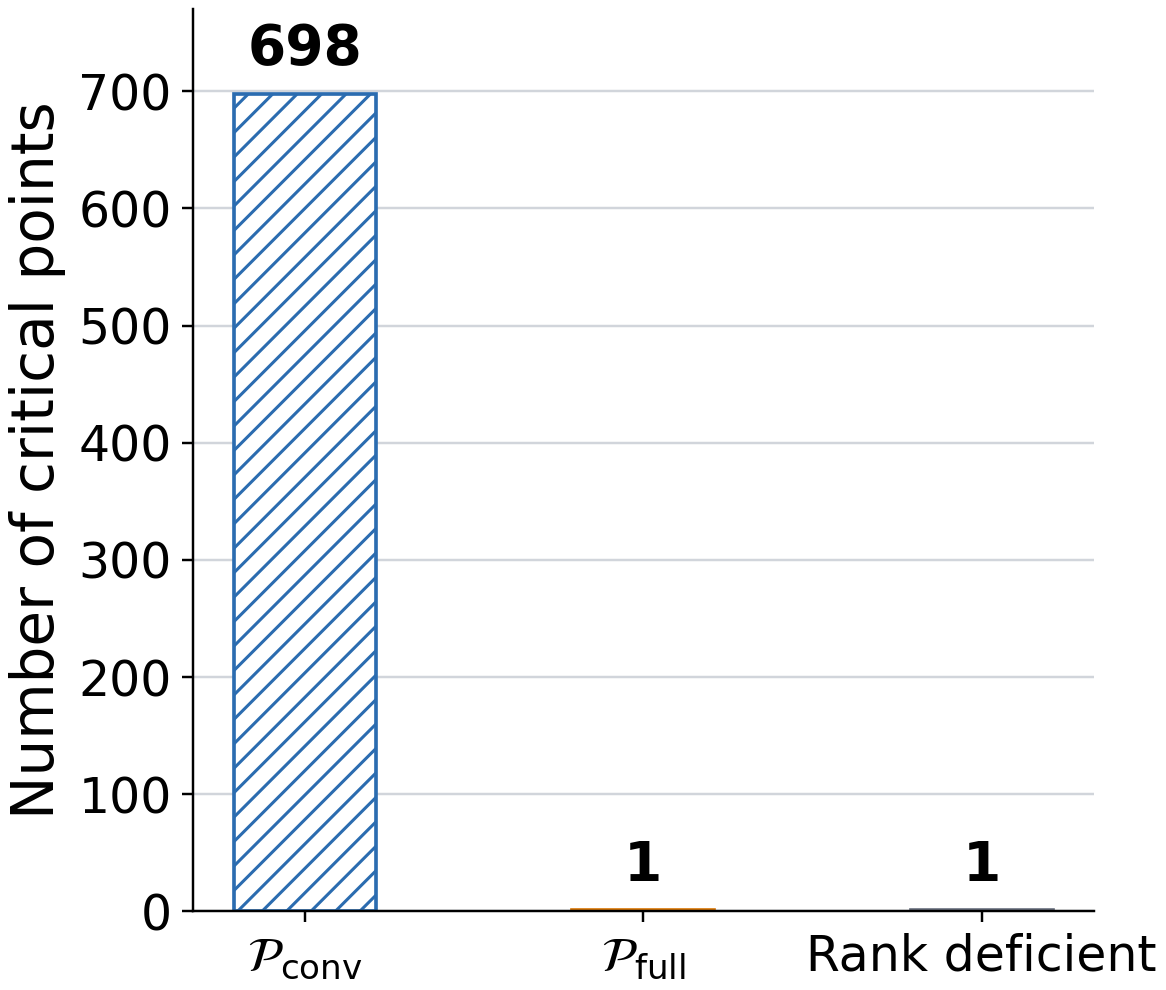}
    \small (c) Layer 3
\end{minipage}
\caption{Task~1: numbers of points in $\mathcal{P}_{\mathrm{conv}}$ and
$\mathcal{P}_{\mathrm{full}}$ for each convolutional target layer.}
\label{fig:task1_layer_type_counts}
\end{figure}

\paragraph{\normalfont\bfseries Task 2: Kernel-Size Identification.}
Figure~\ref{fig:task2_kernel_votes} reports the votes of the points that survive
the boundary and RPCP filtering.
Since the square $\operatorname{Shape}_x$ of every layer concentrates on a
single candidate, the kernel sizes are recovered as $\widehat{k}=3$, $5$, and
$3$ for layers 1--3 by Lemma~\ref{lem:conv_kernel_size_recovery}, matching the
ground truth.

\begin{figure}[!htb]
\centering
\begin{minipage}[t]{0.32\textwidth}
    \centering
    \includegraphics[width=\linewidth]{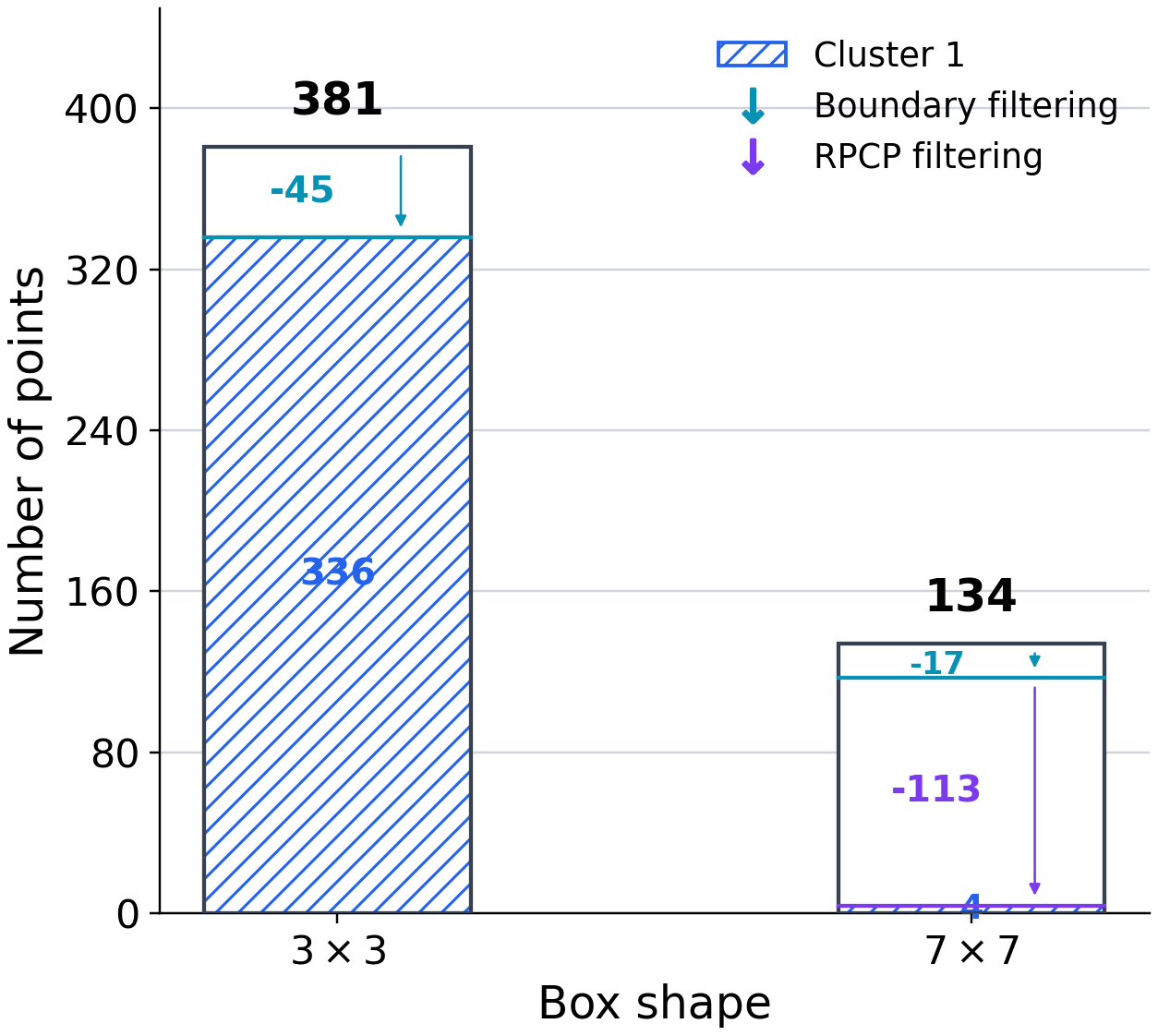}
    \small (a) Layer 1
\end{minipage}
\hfill
\begin{minipage}[t]{0.32\textwidth}
    \centering
    \includegraphics[width=\linewidth]{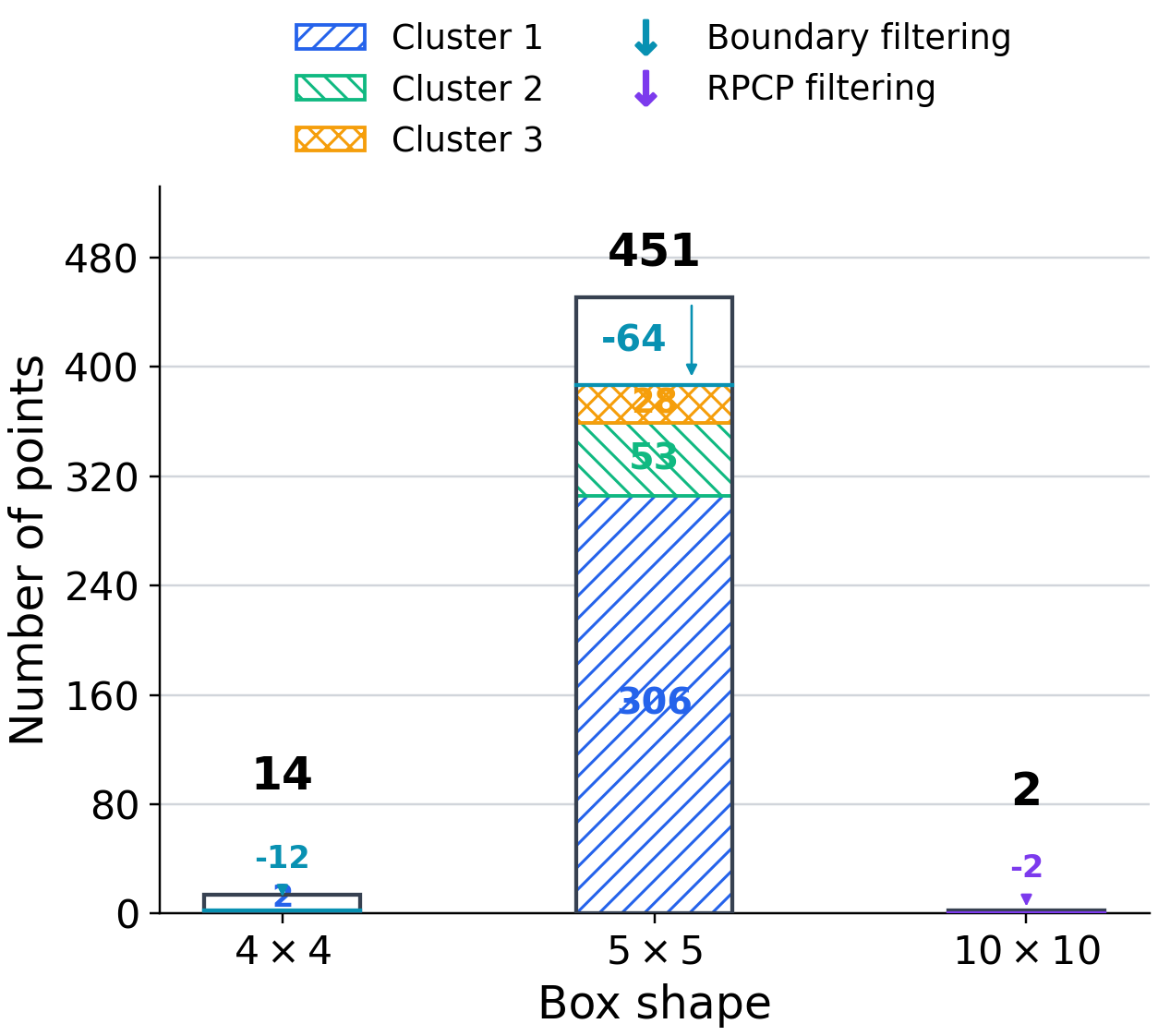}
    \small (b) Layer 2
\end{minipage}
\hfill
\begin{minipage}[t]{0.32\textwidth}
    \centering
    \includegraphics[width=\linewidth]{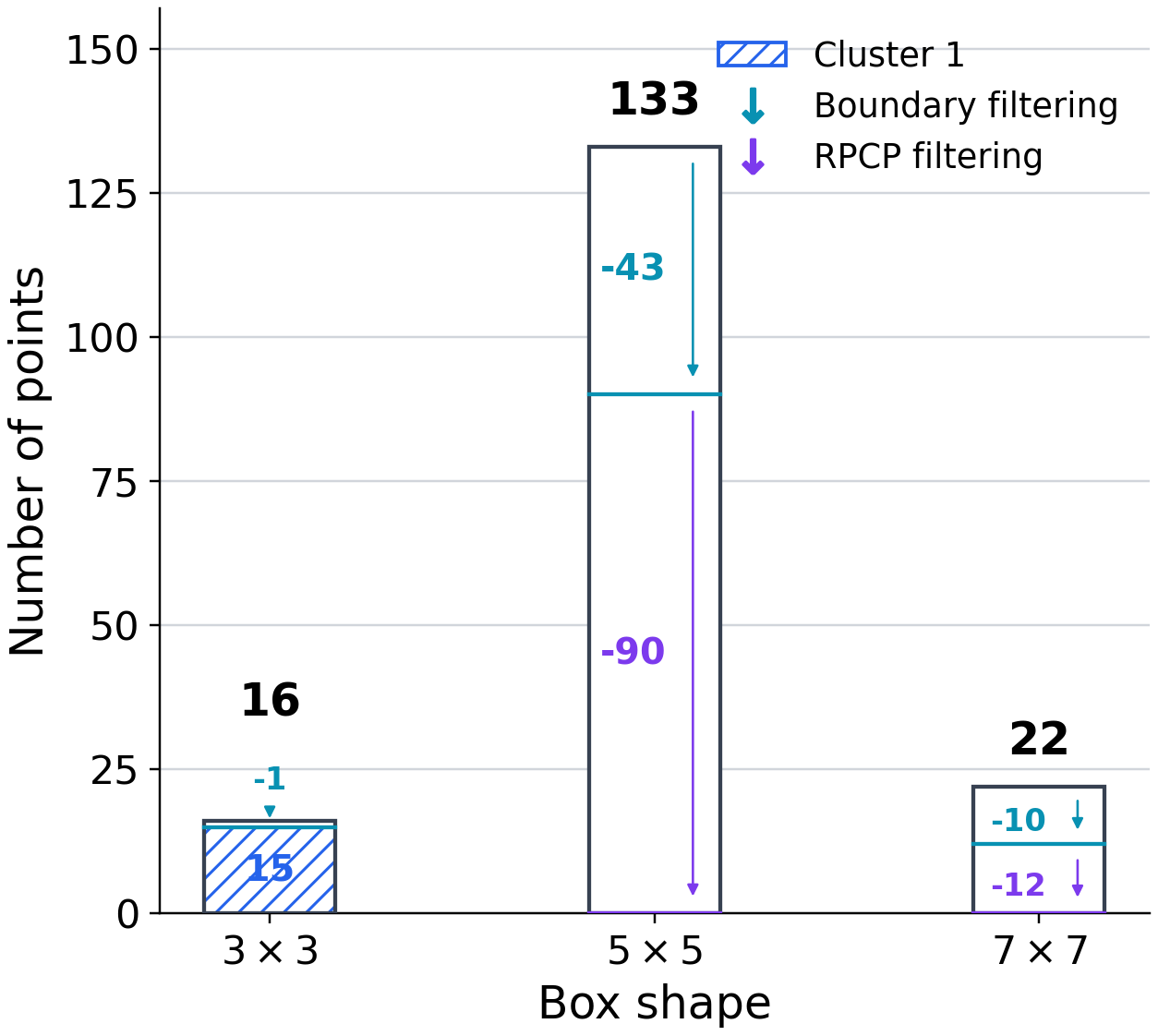}
    \small (c) Layer 3
\end{minipage}
\caption{Task~2: votes of the points that survive the boundary and RPCP
filtering, grouped by candidate kernel size.}
\label{fig:task2_kernel_votes}
\end{figure}

\paragraph{\normalfont\bfseries Task 3: Stride Identification.}
Figure~\ref{fig:task2_window_starts} shows the starting positions of all points
in $\widetilde{\mathcal{P}}_{\mathrm{conv}}^{\circ}$.
Since the minimum difference between the starting positions is $1$,
$1$, and $2$ for layers 1--3, the strides are recovered by
Lemma~\ref{lem:conv_window_start_recovery}, matching the ground truth.

\begin{figure}[!htb]
\centering
\begin{minipage}[t]{0.32\textwidth}
    \centering
    \includegraphics[width=\linewidth]{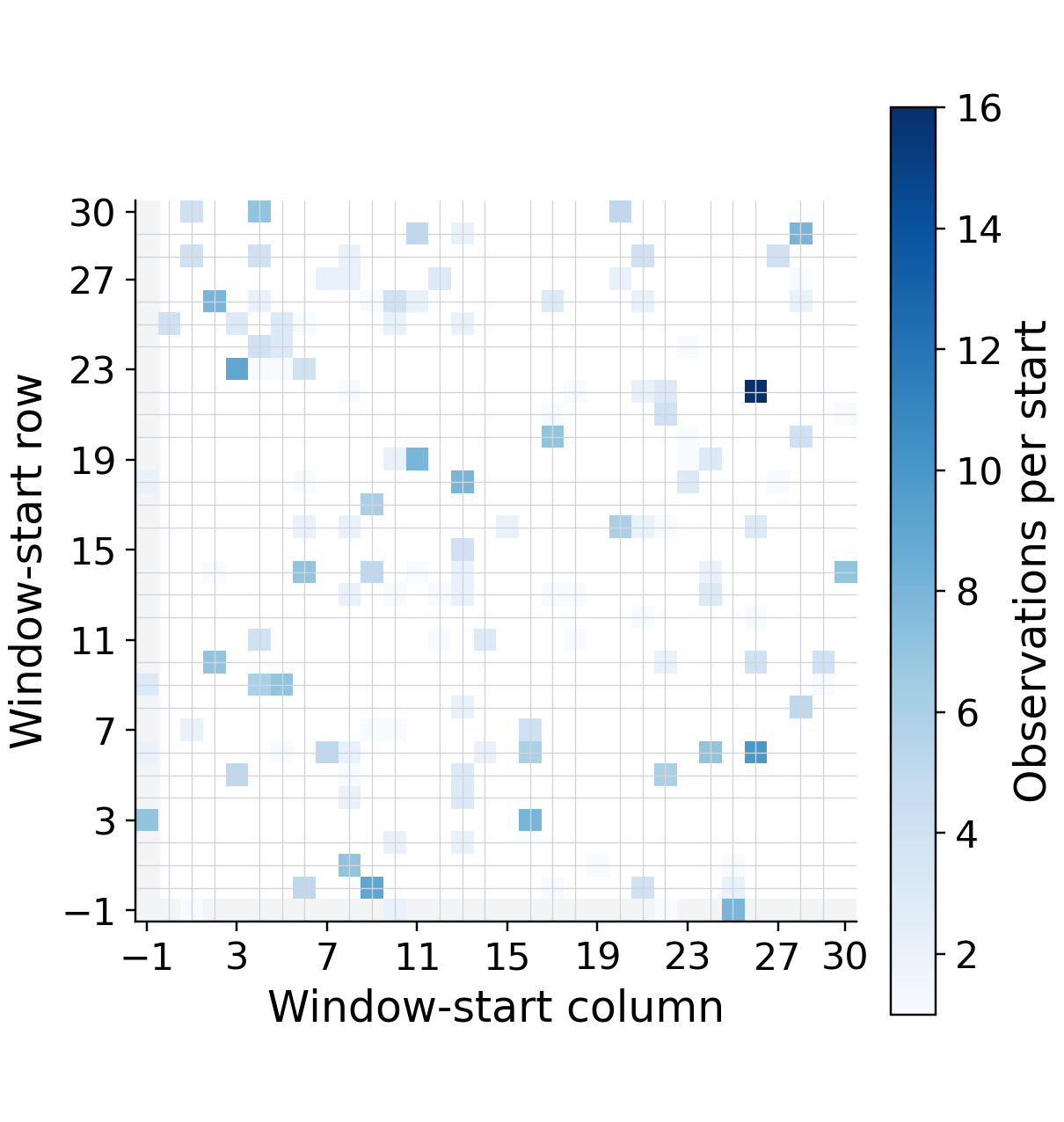}
    \small (a) Layer 1
\end{minipage}
\hfill
\begin{minipage}[t]{0.32\textwidth}
    \centering
    \includegraphics[width=\linewidth]{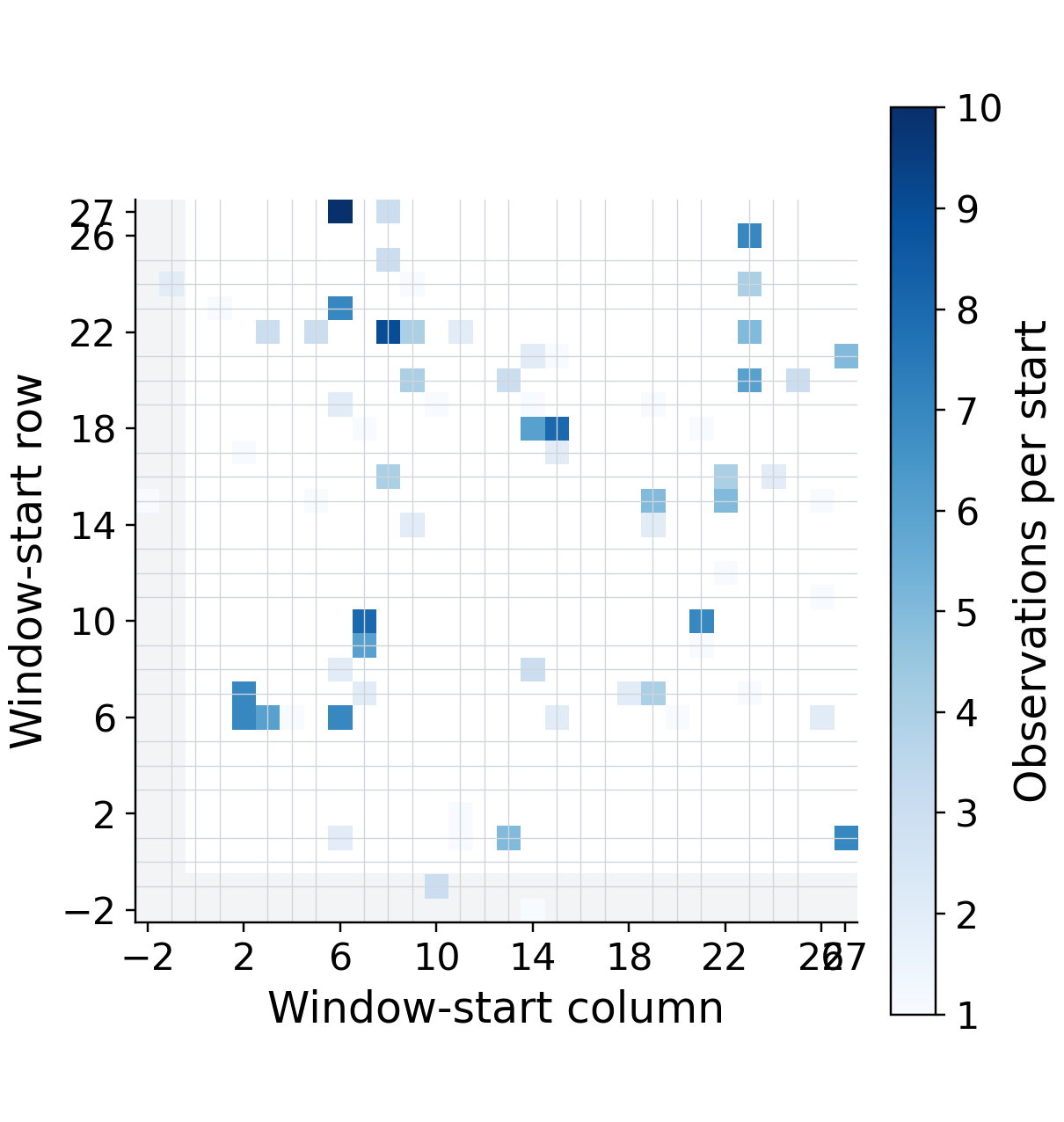}
    \small (b) Layer 2
\end{minipage}
\hfill
\begin{minipage}[t]{0.32\textwidth}
    \centering
    \includegraphics[width=\linewidth]{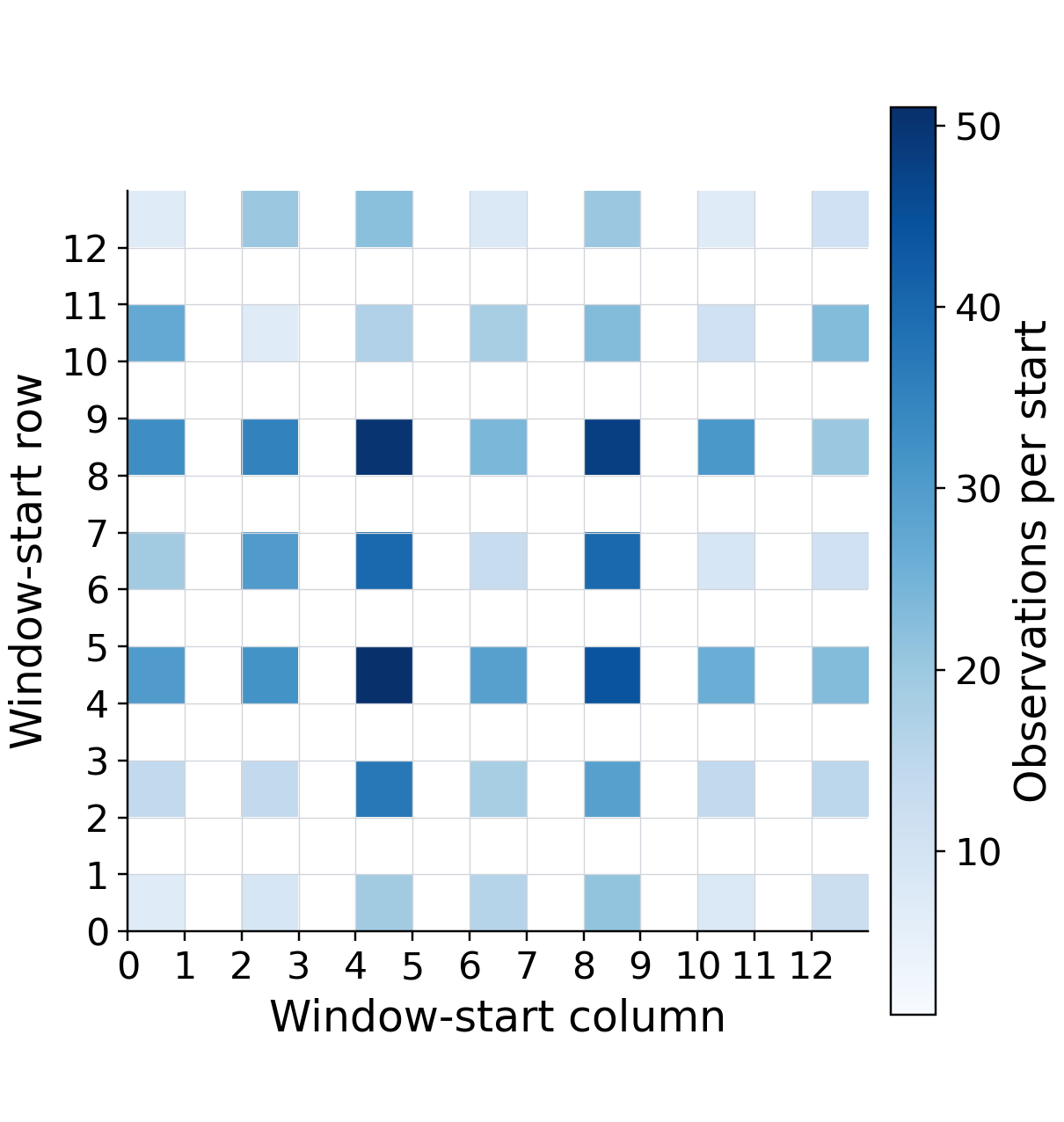}
    \small (c) Layer 3
\end{minipage}
\caption{Task~3: window-start distributions of the points in
$\widetilde{\mathcal{P}}_{\mathrm{conv}}^{\circ}$.}
\label{fig:task2_window_starts}
\end{figure}

\paragraph{\normalfont\bfseries Task 4: Padding-Mode Identification.}
Figure~\ref{fig:task3_same_evidence} shows the truncated points identified on the
four boundaries of each layer.
Since truncated points appear on all four boundaries in layers 1 and 2 and on
none of them in layer 3, the padding modes are recovered as Same, Same, and
Valid for layers 1--3 by Lemma~\ref{lem:padding_boundary}, matching the ground
truth.

\begin{figure}[!htb]
\centering
\begin{minipage}[t]{0.32\textwidth}
    \centering
    \includegraphics[width=\linewidth]{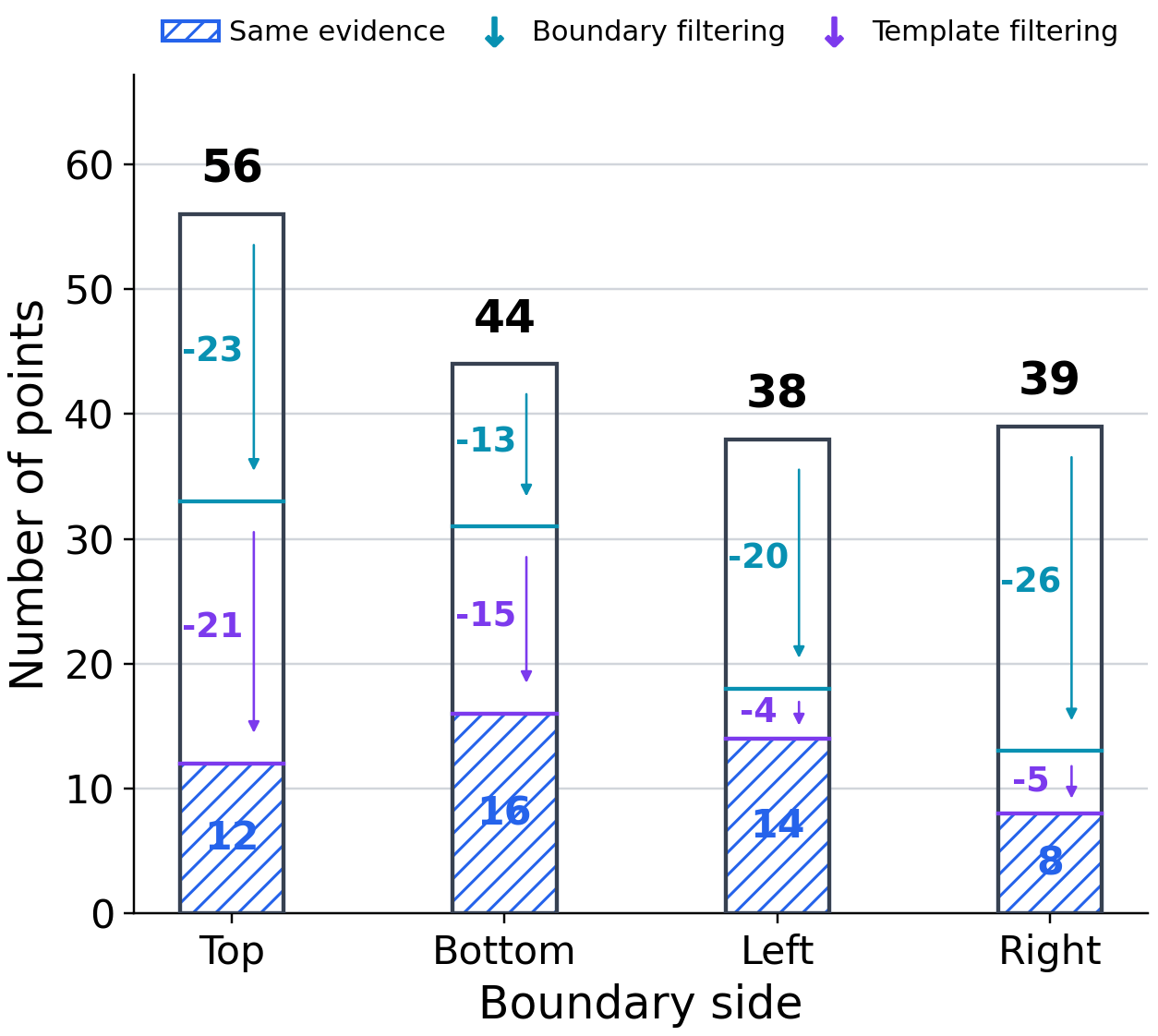}
    \small (a) Layer 1
\end{minipage}
\hfill
\begin{minipage}[t]{0.32\textwidth}
    \centering
    \includegraphics[width=\linewidth]{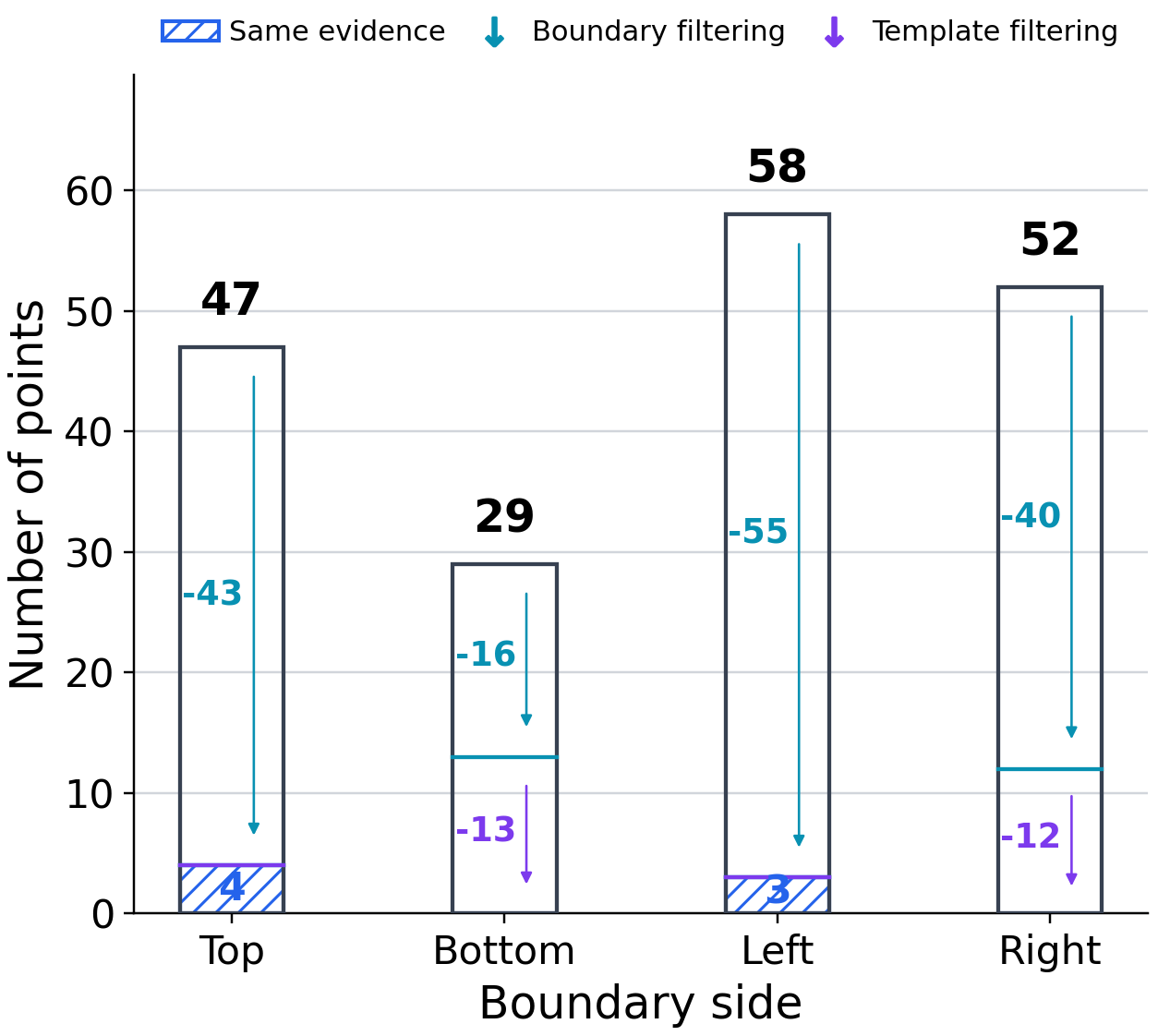}
    \small (b) Layer 2
\end{minipage}
\hfill
\begin{minipage}[t]{0.32\textwidth}
    \centering
    \includegraphics[width=\linewidth]{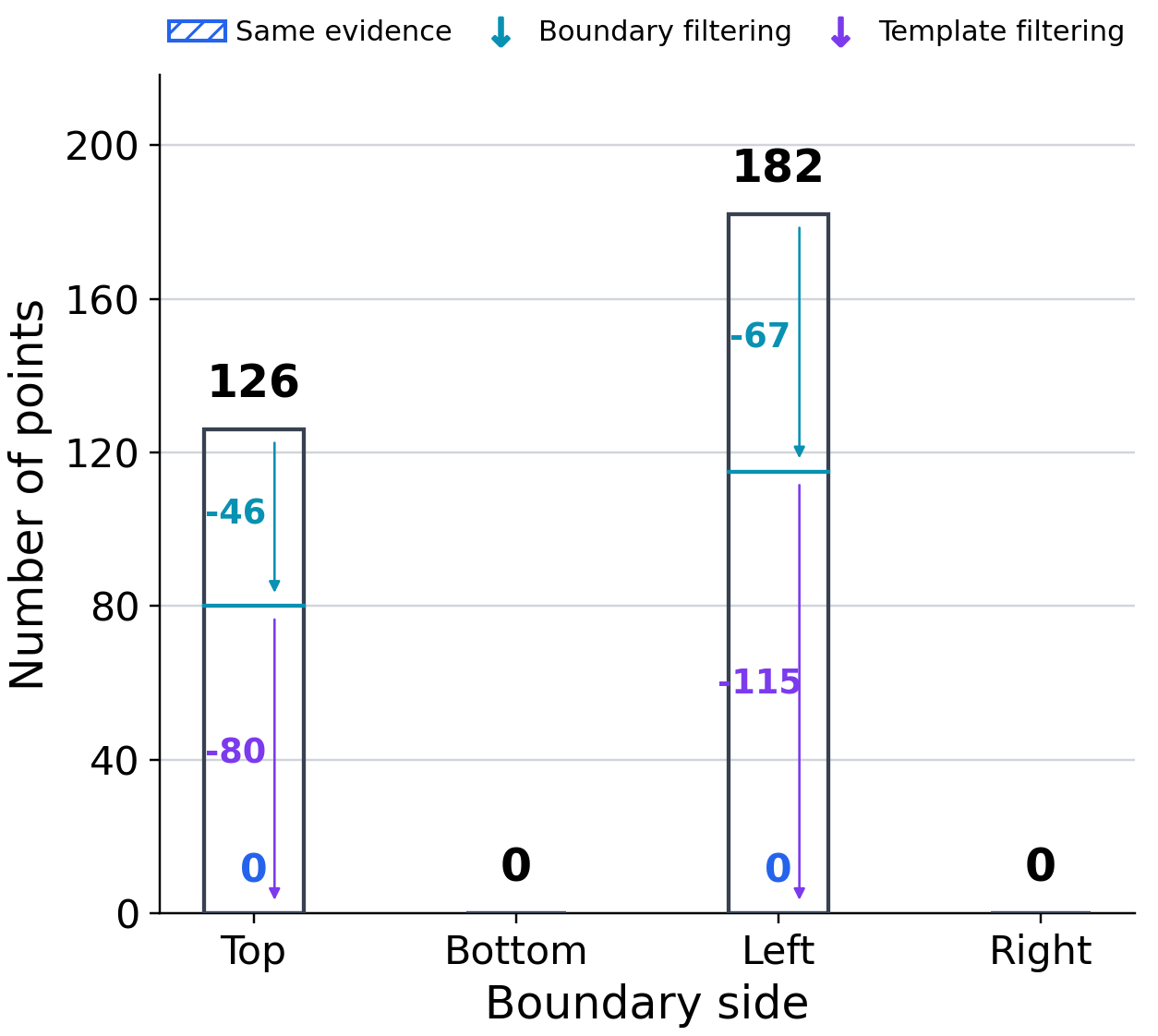}
    \small (c) Layer 3
\end{minipage}
\caption{Task~4: the truncated points on the four boundaries of each layer. Cyan
arrows mark the points removed by the boundary filtering, and purple arrows mark
the points that the template filtering rejects.}
\label{fig:task3_same_evidence}
\end{figure}

\paragraph{\normalfont\bfseries Task 5: Output-Channel Identification.}
Figure~\ref{fig:task4_output_channels} shows how the RPCP and PSP records of each
layer are grouped into clusters.
Since the records of each layer form $1$, $2$, and $3$ clusters, the numbers of
output channels are recovered,
matching the ground truth.

\begin{figure}[!htb]
\centering
\begin{minipage}[t]{0.32\textwidth}
    \centering
    \includegraphics[width=\linewidth]{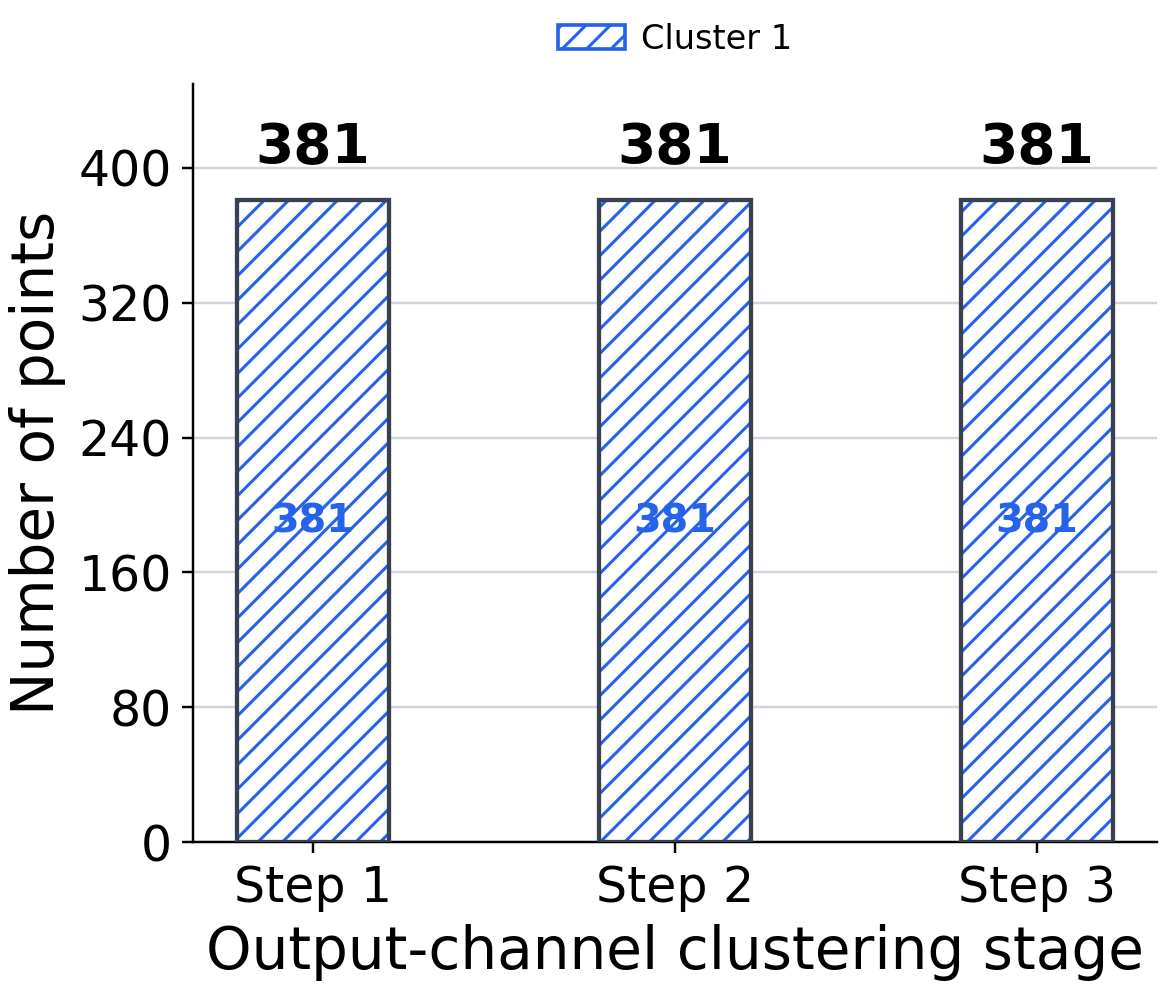}
    \small (a) Layer 1
\end{minipage}
\hfill
\begin{minipage}[t]{0.32\textwidth}
    \centering
    \includegraphics[width=\linewidth]{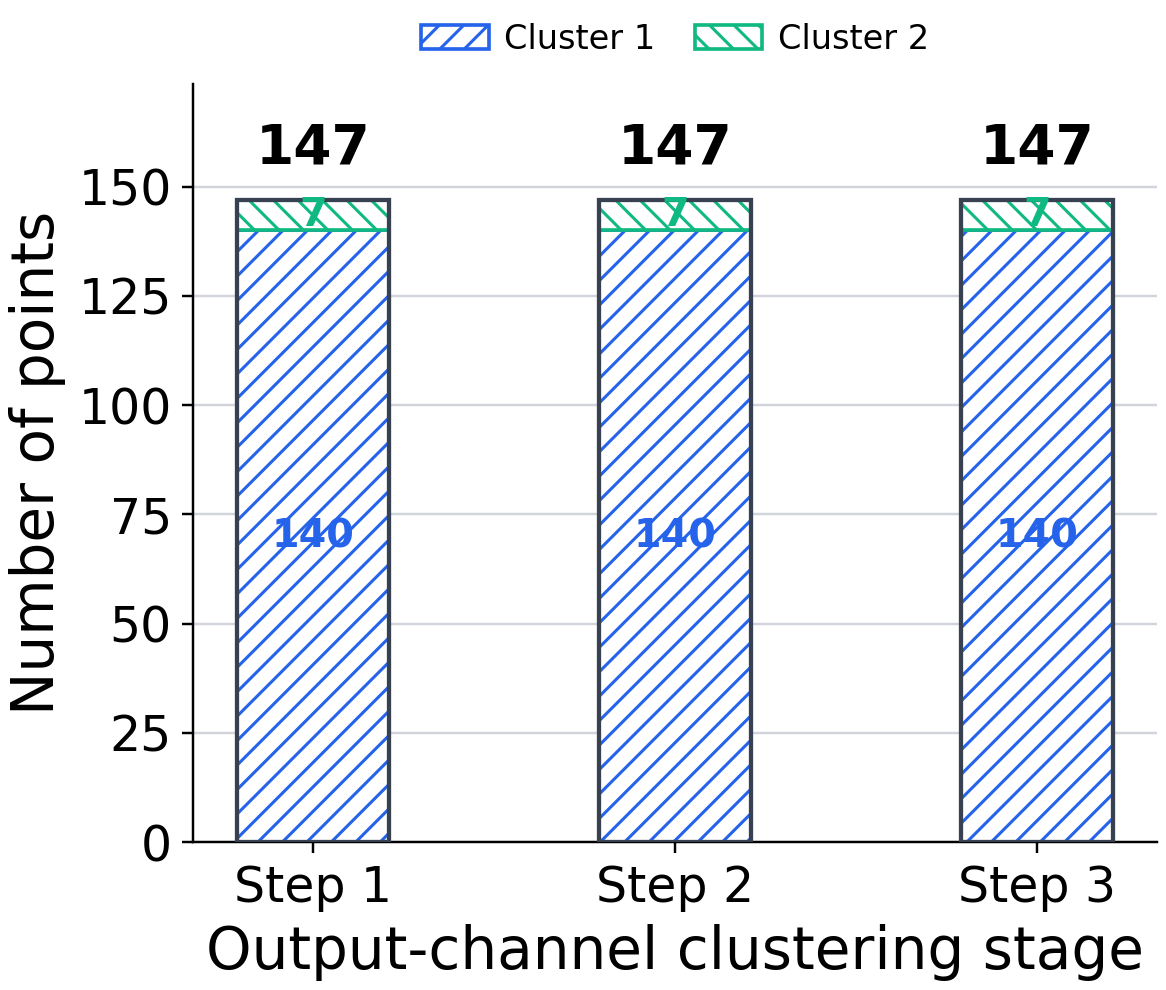}
    \small (b) Layer 2
\end{minipage}
\hfill
\begin{minipage}[t]{0.32\textwidth}
    \centering
    \includegraphics[width=\linewidth]{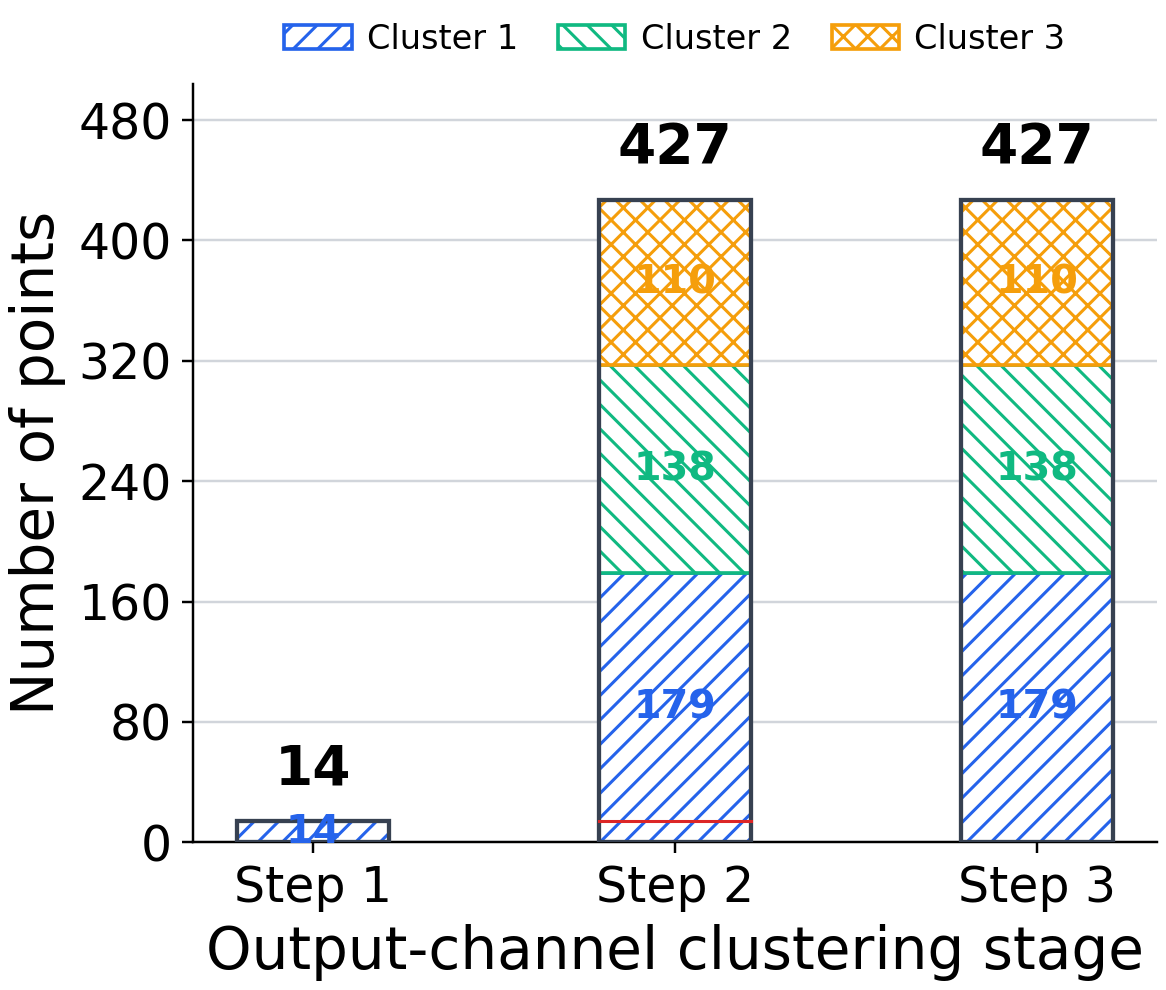}
    \small (c) Layer 3
\end{minipage}
\caption{Task~5: clustering of the RPCP and PSP records. Step~1 clusters the
stable RPCP ratio records, Step~2 adds the PSPs whose numerical ratios match the
cluster centers, and Step~3 adds a new cluster for the remaining PSPs.}
\label{fig:task4_output_channels}
\end{figure}

\paragraph{\normalfont\bfseries Task 6: Pooling-Structure Identification.}
Figure~\ref{fig:task5_pooling_identification} shows the $\widetilde{v}_x^{(L+1)}$ of
layers 1 and 2 and the max-pooling parameter evidence of layer 3. Layer 1 shows
no repeated blocks, so $\widehat{\rho}^{(1)} = \mathrm{NoPool}$. Layer 2
shows repeated constant blocks, so $\widehat{\rho}^{(2)} = \mathrm{AvgPool}$. The pooling parameters are
identified by Corollary~\ref{cor:avg_pool_parameter_estimate}. For layer 3, PSPs identify $\widehat{\rho}^{(3)} = \mathrm{MaxPool}$, and the anchor positions
$\mathcal{R}_{\max}$ and $\mathcal{V}_{\max}$ give the starting points, from which the pooling window size and the stride follow, by
Proposition~\ref{prop:psp_window_start_offset} and Proposition~\ref{prop:max_pooling_stride}.
The pooling structures are recovered as NoPool, AvgPool $(2,2)$, and
MaxPool $(3,2)$ for layers 1--3, matching the ground truth.

%\begingroup
%\setlength{\intextsep}{14pt}
\begin{figure}[!htb]
\centering
\begin{minipage}[t]{0.32\textwidth}
\centering
\includegraphics[width=\linewidth]{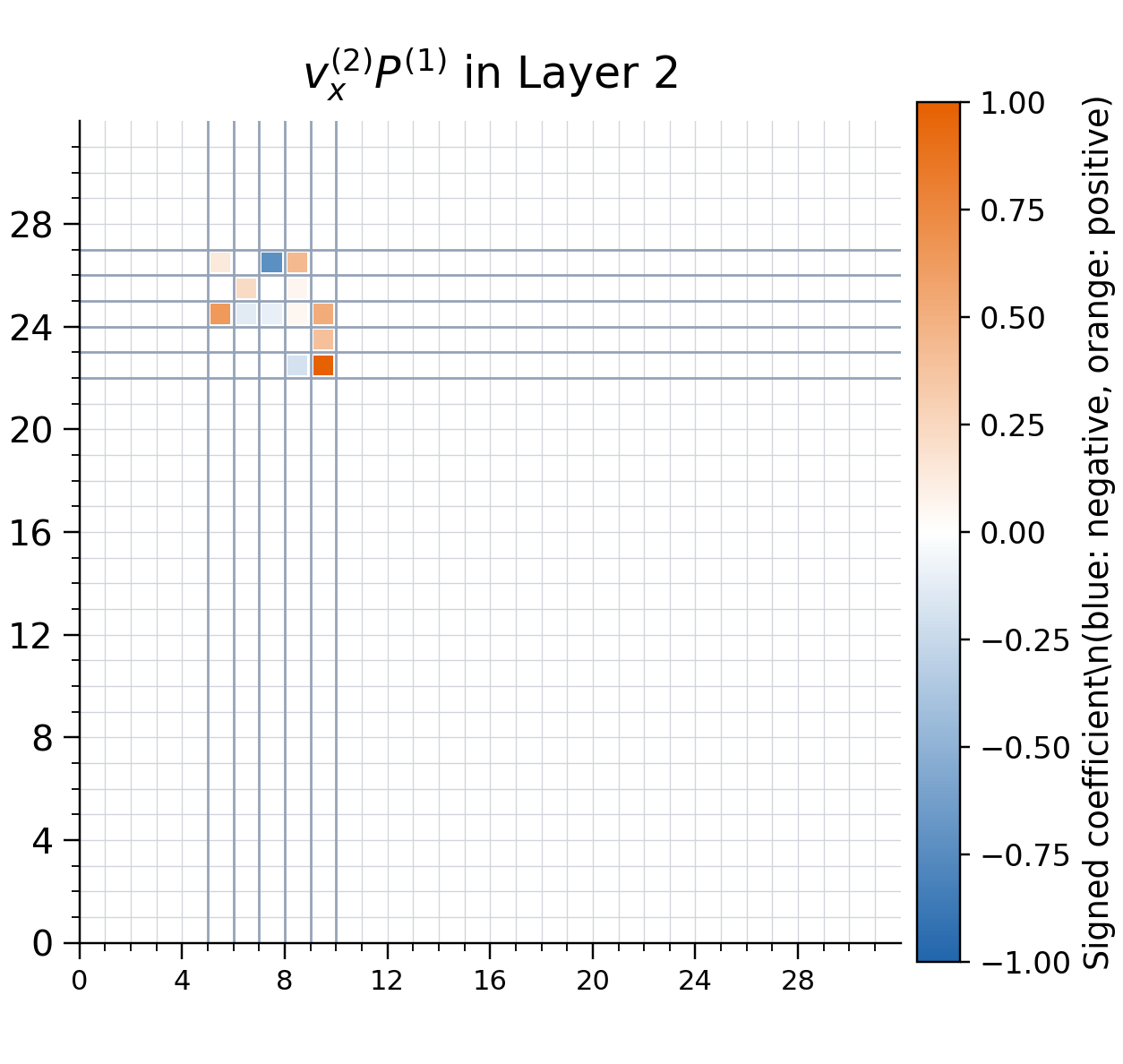}
\small (a) Layer 1
\end{minipage}\hfill
\begin{minipage}[t]{0.32\textwidth}
\centering
\includegraphics[width=\linewidth]{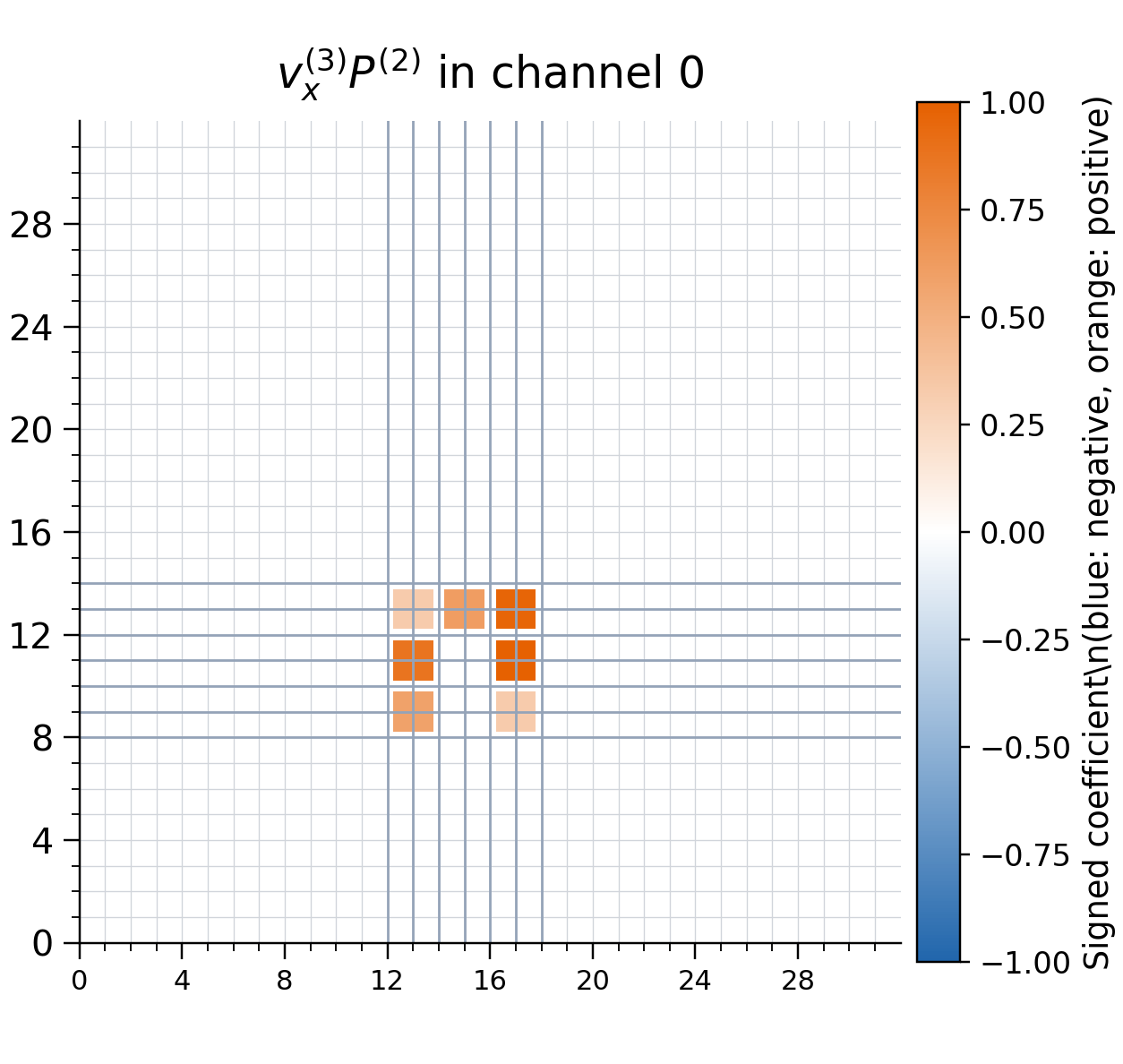}
\small (b) Layer 2
\end{minipage}\hfill
\begin{minipage}[t]{0.32\textwidth}
\centering
\includegraphics[width=\linewidth]{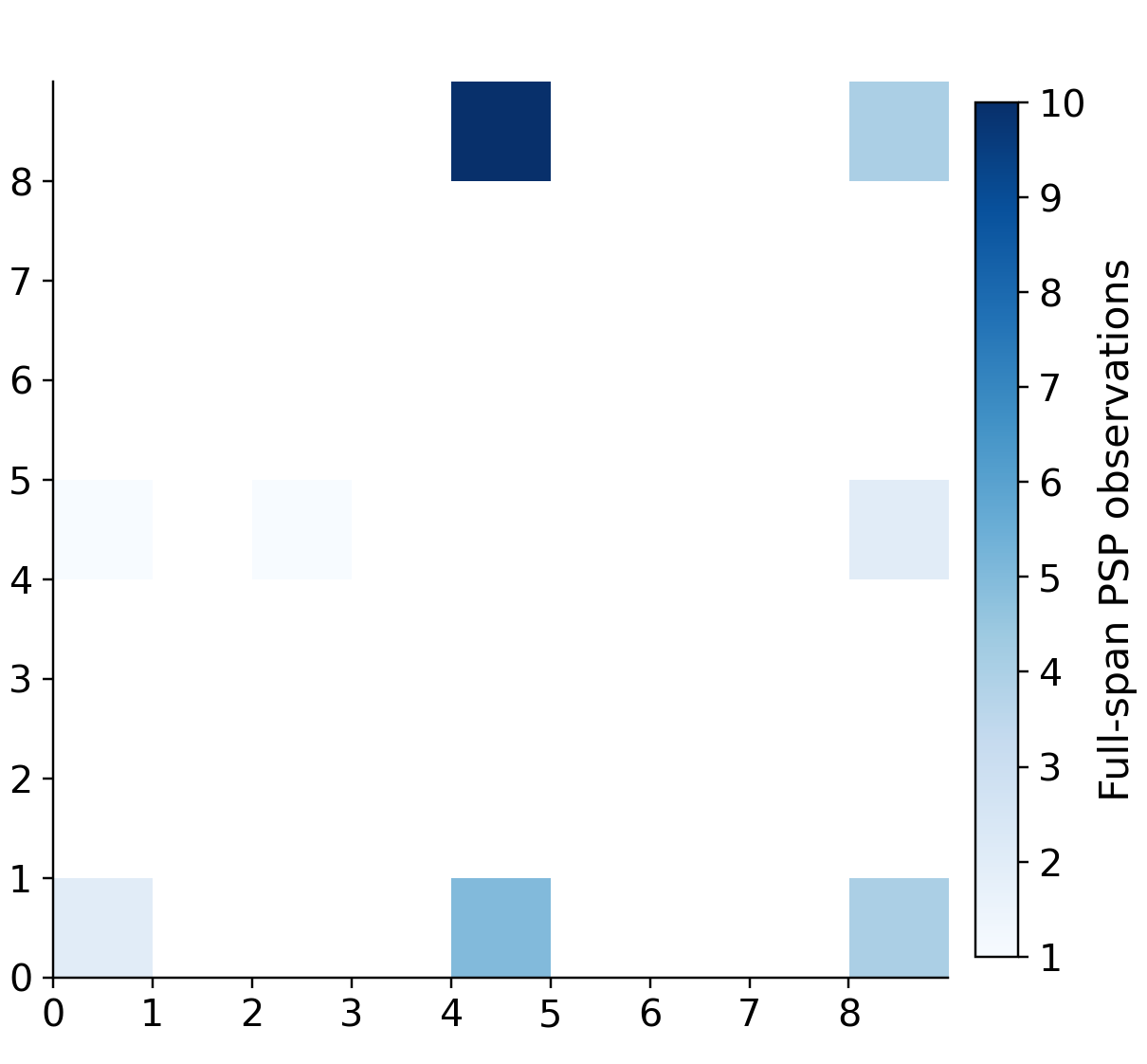}
\small (c) Layer 3
\end{minipage}
\caption{Task~6: Panel (a) shows a representative restored vector of layer 1
(up to the factor $p^2$), panel (b) the restored vector of layer 2 (channel 0),
and panel (c) the max-pooling parameter evidence of layer 3.}
\label{fig:task5_pooling_identification}
\end{figure}

% \vspace{-0.2cm}
\section{Conclusions}

In this paper, we initiated the study of architecture recovery for
convolutional neural networks. Focusing on the convolutional layers, we
proposed an attack that recovers the architecture of each layer before its
parameters and, combined with existing methods for the fully
connected layers, completes a cryptanalytic extraction framework for CNNs.
Our approach is built on a direct observation: the recovered
vectors reveal the network architecture in their spatial structure.
We identified three such signals, the
sparsity consistency with the convolution receptive field, the numerical
consistency with the kernel parameters, and the structural consistency with
the pooling operation, and used them to recover the layer type, the kernel
size and stride, the padding mode, the output-channel number, and the pooling
structure.
Experiments on
a range of CNNs, including a purpose-built model with diverse
architectural hyperparameters and the models
adopted by existing attacks, verified the effectiveness of our attack.

Additionally, our attack complements the established architecture-recovery
results for fully connected networks: existing methods recover the layer
widths of a fully connected network along with its parameters, whereas we
show that the architecture of the convolutional layers of a CNN is
directly revealed by the spatial geometry of the recovered weight vectors.
The two routes cover the fully connected and the convolutional layers,
respectively, and together they remove the architectural prior from the
cryptanalytic extraction of CNNs.
\appendix
%\input{data/10-appendix}

% \section*{Acknowledgement}
% This work is supported by the National Key Research and Development Program of China 
% (2018YFB0803405, 2017YFA0303903). 

\clearpage
\bibliographystyle{splncs04}
\bibliography{reference}

@inproceedings{DBLP:journals/iacr/ItoMT25,
  author       = {Akira Ito and
                  Takayuki Miura and
                  Yosuke Todo},
  editor       = {Nadia Heninger and
                  Mike Rosulek},
  title        = {Is the Hard-Label Cryptanalytic Model Extraction Really Polynomial?},
  booktitle    = {Advances in Cryptology - {CRYPTO} 2026 - 46th Annual International
                  Cryptology Conference, Santa Barbara, CA, USA, August 17-20, 2026,
                  Proceedings, Part {VII}},
  series       = {Lecture Notes in Computer Science},
  volume       = {16806},
  pages        = {67--98},
  publisher    = {Springer},
  year         = {2026},
  url          = {https://doi.org/10.1007/978-3-032-35415-0\_3},
  doi          = {10.1007/978-3-032-35415-0\_3}
}

@inproceedings{DBLP:conf/asiacrypt/ChenDMSWYW25,
  author       = {Yi Chen and
                  Xiaoyang Dong and
                  Ruijie Ma and
                  Yantian Shen and
                  Anyu Wang and
                  Hongbo Yu and
                  Xiaoyun Wang},
  editor       = {Goichiro Hanaoka and
                  Bo{-}Yin Yang},
  title        = {Delving into Cryptanalytic Extraction of PReLU Neural Networks},
  booktitle    = {Advances in Cryptology - {ASIACRYPT} 2025 - 31st International Conference
                  on the Theory and Application of Cryptology and Information Security,
                  Melbourne, VIC, Australia, December 8-12, 2025, Proceedings, Part
                  {II}},
  series       = {Lecture Notes in Computer Science},
  volume       = {16246},
  pages        = {576--607},
  publisher    = {Springer},
  year         = {2025},
  url          = {https://doi.org/10.1007/978-981-95-5096-8\_18},
  doi          = {10.1007/978-981-95-5096-8\_18},
  bibsource    = {dblp computer science bibliography, https://dblp.org}
}

@article{DBLP:journals/tnn/Baum91,
  author       = {Eric B. Baum},
  title        = {Neural net algorithms that learn in polynomial time from examples
                  and queries},
  journal      = {{IEEE} Trans. Neural Networks},
  volume       = {2},
  number       = {1},
  pages        = {5--19},
  year         = {1991},
  url          = {https://doi.org/10.1109/72.80287},
  doi          = {10.1109/72.80287},
  bibsource    = {dblp computer science bibliography, https://dblp.org}
}

@inproceedings{DBLP:conf/uss/TramerZJRR16,
  author       = {Florian Tram{\`{e}}r and
                  Fan Zhang and
                  Ari Juels and
                  Michael K. Reiter and
                  Thomas Ristenpart},
  editor       = {Thorsten Holz and
                  Stefan Savage},
  title        = {Stealing Machine Learning Models via Prediction APIs},
  booktitle    = {25th {USENIX} Security Symposium, {USENIX} Security 16, Austin, TX,
                  USA, August 10-12, 2016},
  pages        = {601--618},
  publisher    = {{USENIX} Association},
  year         = {2016},
  url          = {https://www.usenix.org/conference/usenixsecurity16/technical-sessions/presentation/tramer},
  bibsource    = {dblp computer science bibliography, https://dblp.org}
}

@inproceedings{DBLP:conf/uss/JagielskiCBKP20,
  author       = {Matthew Jagielski and
                  Nicholas Carlini and
                  David Berthelot and
                  Alex Kurakin and
                  Nicolas Papernot},
  editor       = {Srdjan Capkun and
                  Franziska Roesner},
  title        = {High Accuracy and High Fidelity Extraction of Neural Networks},
  booktitle    = {29th {USENIX} Security Symposium, {USENIX} Security 2020, August 12-14,
                  2020},
  pages        = {1345--1362},
  publisher    = {{USENIX} Association},
  year         = {2020},
  url          = {https://www.usenix.org/conference/usenixsecurity20/presentation/jagielski},
  bibsource    = {dblp computer science bibliography, https://dblp.org}
}

@inproceedings{DBLP:conf/icml/RolnickK20,
  author       = {David Rolnick and
                  Konrad P. Kording},
  title        = {Reverse-engineering deep ReLU networks},
  booktitle    = {Proceedings of the 37th International Conference on Machine Learning,
                  {ICML} 2020, 13-18 July 2020, Virtual Event},
  series       = {Proceedings of Machine Learning Research},
  volume       = {119},
  pages        = {8178--8187},
  publisher    = {{PMLR}},
  year         = {2020},
  url          = {http://proceedings.mlr.press/v119/rolnick20a.html},
  bibsource    = {dblp computer science bibliography, https://dblp.org}
}

@inproceedings{DBLP:conf/iclr/DanielyG23,
  author       = {Amit Daniely and
                  Elad Granot},
  title        = {An Exact Poly-Time Membership-Queries Algorithm for Extracting a Three-Layer
                  ReLU Network},
  booktitle    = {The Eleventh International Conference on Learning Representations,
                  {ICLR} 2023, Kigali, Rwanda, May 1-5, 2023},
  publisher    = {OpenReview.net},
  year         = {2023},
  url          = {https://openreview.net/forum?id=-CoNloheTs},
  bibsource    = {dblp computer science bibliography, https://dblp.org}
}

@inproceedings{DBLP:conf/crypto/CarliniJM20,
  author       = {Nicholas Carlini and
                  Matthew Jagielski and
                  Ilya Mironov},
  editor       = {Daniele Micciancio and
                  Thomas Ristenpart},
  title        = {Cryptanalytic Extraction of Neural Network Models},
  booktitle    = {Advances in Cryptology - {CRYPTO} 2020 - 40th Annual International
                  Cryptology Conference, {CRYPTO} 2020, Santa Barbara, CA, USA, August
                  17-21, 2020, Proceedings, Part {III}},
  series       = {Lecture Notes in Computer Science},
  volume       = {12172},
  pages        = {189--218},
  publisher    = {Springer},
  year         = {2020},
  url          = {https://doi.org/10.1007/978-3-030-56877-1\_7},
  doi          = {10.1007/978-3-030-56877-1\_7},
  bibsource    = {dblp computer science bibliography, https://dblp.org}
}

@inproceedings{DBLP:conf/eurocrypt/CanalesMartinezCHRSS24,
  author       = {Isaac Andr{\'{e}}s Canales Martinez and
                  Jorge Ch{\'{a}}vez{-}Saab and
                  Anna Hambitzer and
                  Francisco Rodr{\'{\i}}guez{-}Henr{\'{\i}}quez and
                  Nitin Satpute and
                  Adi Shamir},
  editor       = {Marc Joye and
                  Gregor Leander},
  title        = {Polynomial Time Cryptanalytic Extraction of Neural Network Models},
  booktitle    = {Advances in Cryptology - {EUROCRYPT} 2024 - 43rd Annual International
                  Conference on the Theory and Applications of Cryptographic Techniques,
                  Zurich, Switzerland, May 26-30, 2024, Proceedings, Part {III}},
  series       = {Lecture Notes in Computer Science},
  volume       = {14653},
  pages        = {3--33},
  publisher    = {Springer},
  year         = {2024},
  url          = {https://doi.org/10.1007/978-3-031-58734-4\_1},
  doi          = {10.1007/978-3-031-58734-4\_1},
  bibsource    = {dblp computer science bibliography, https://dblp.org}
}

@inproceedings{DBLP:conf/asiacrypt/ChenDGSWW24,
  author       = {Yi Chen and
                  Xiaoyang Dong and
                  Jian Guo and
                  Yantian Shen and
                  Anyu Wang and
                  Xiaoyun Wang},
  editor       = {Kai{-}Min Chung and
                  Yu Sasaki},
  title        = {Hard-Label Cryptanalytic Extraction of Neural Network Models},
  booktitle    = {Advances in Cryptology - {ASIACRYPT} 2024 - 30th International Conference
                  on the Theory and Application of Cryptology and Information Security,
                  Kolkata, India, December 9-13, 2024, Proceedings, Part {VIII}},
  series       = {Lecture Notes in Computer Science},
  volume       = {15491},
  pages        = {207--236},
  publisher    = {Springer},
  year         = {2024},
  url          = {https://doi.org/10.1007/978-981-96-0944-4\_7},
  doi          = {10.1007/978-981-96-0944-4\_7},
  bibsource    = {dblp computer science bibliography, https://dblp.org}
}

@inproceedings{DBLP:conf/eurocrypt/CarliniCHRS25,
  author       = {Nicholas Carlini and
                  Jorge Ch{\'{a}}vez{-}Saab and
                  Anna Hambitzer and
                  Francisco Rodr{\'{\i}}guez{-}Henr{\'{\i}}quez and
                  Adi Shamir},
  editor       = {Serge Fehr and
                  Pierre{-}Alain Fouque},
  title        = {Polynomial Time Cryptanalytic Extraction of Deep Neural Networks in
                  the Hard-Label Setting},
  booktitle    = {Advances in Cryptology - {EUROCRYPT} 2025 - 44th Annual International
                  Conference on the Theory and Applications of Cryptographic Techniques,
                  Madrid, Spain, May 4-8, 2025, Proceedings, Part {I}},
  series       = {Lecture Notes in Computer Science},
  volume       = {15601},
  pages        = {364--396},
  publisher    = {Springer},
  year         = {2025},
  url          = {https://doi.org/10.1007/978-3-031-91107-1\_13},
  doi          = {10.1007/978-3-031-91107-1\_13},
  bibsource    = {dblp computer science bibliography, https://dblp.org}
}

@inproceedings{DBLP:conf/nips/FoersterMSH24,
  author       = {Hanna Foerster and
                  Robert D. Mullins and
                  Ilia Shumailov and
                  Jamie Hayes},
  editor       = {Amir Globerson and
                  Lester Mackey and
                  Danielle Belgrave and
                  Angela Fan and
                  Ulrich Paquet and
                  Jakub M. Tomczak and
                  Cheng Zhang},
  title        = {Beyond Slow Signs in High-fidelity Model Extraction},
  booktitle    = {Advances in Neural Information Processing Systems 38: Annual Conference
                  on Neural Information Processing Systems 2024, NeurIPS 2024, Vancouver,
                  BC, Canada, December 10 - 15, 2024},
  year         = {2024},
  url          = {http://papers.nips.cc/paper\_files/paper/2024/hash/22ae669a35bb9e70eb93ab77c1eff5b4-Abstract-Conference.html},
  bibsource    = {dblp computer science bibliography, https://dblp.org}
}

@inproceedings{DBLP:conf/eurocrypt/LiuSELBP26,
  author       = {Haolin Liu and
                  Adrien Siproudhis and
                  Samuel Experton and
                  Peter Lorenz and
                  Christina Boura and
                  Thomas Peyrin},
  title        = {Navigating the Deep: End-to-End Extraction on Deep Neural Networks},
  booktitle    = {Advances in Cryptology - {EUROCRYPT} 2026},
  editor       = {Joachim Daemen and
                  Emmanuel Thom{\'e}},
  series       = {Lecture Notes in Computer Science},
  volume       = {16546},
  publisher    = {Springer},
  year         = {2026},
  doi          = {10.1007/978-3-032-25333-0\_17},
  url          = {https://doi.org/10.1007/978-3-032-25333-0_17}
}

@inproceedings{DBLP:conf/latincrypt/CanalesMartinezS25,
  author       = {Isaac A. Canales{-}Mart{\'{\i}}nez and
                  David Santos},
  editor       = {Daniel Escudero and
                  Ivan Damg{\aa}rd},
  title        = {Extracting Some Layers of Deep Neural Networks in the Hard-Label Setting},
  booktitle    = {Progress in Cryptology - {LATINCRYPT} 2025 - 9th International Conference
                  on Cryptology and Information Security in Latin America, Medell{\'{\i}}n,
                  Colombia, October 1-3, 2025, Proceedings},
  series       = {Lecture Notes in Computer Science},
  volume       = {16129},
  pages        = {399--421},
  publisher    = {Springer},
  year         = {2025},
  url          = {https://doi.org/10.1007/978-3-032-06754-8\_15},
  doi          = {10.1007/978-3-032-06754-8\_15},
  bibsource    = {dblp computer science bibliography, https://dblp.org}
}

@misc{cryptoeprint:2026/139,
      author = {Xiaohan Sun and Hao Lei and Longxiang Wei and Xiaokang Qi and Kai Hu and Meiqin Wang and Wei Wang},
      title = {Cryptanalytic Extraction of Convolutional Neural Networks},
      howpublished = {Cryptology {ePrint} Archive, Paper 2026/139},
      year = {2026},
      url = {https://eprint.iacr.org/2026/139}
}

@article{hal-05573262,
  author       = {Haolin Liu and Adrien Siproudhis and Christina Boura and Thomas Peyrin},
  title        = {Model Extraction of Convolutional Neural Networks with Max-Pooling},
  journal      = {IACR Transactions on Symmetric Cryptology ({ToSC})},
  year         = {2026},
  url          = {https://eprint.iacr.org/2026/464}
}

@inproceedings{cryptoeprint:2026/241,
  author       = {Zirui Chen and
                  Shi Tang and
                  Zhengchao Gao and
                  Yongjia Su and
                  Lingyue Qin and
                  Xiaoyang Dong},
  editor       = {Nadia Heninger and
                  Mike Rosulek},
  title        = {Algebraic Attack on Convolutional Neural Networks with Max Pooling},
  booktitle    = {Advances in Cryptology - {CRYPTO} 2026 - 46th Annual International
                  Cryptology Conference, Santa Barbara, CA, USA, August 17-20, 2026,
                  Proceedings, Part {VII}},
  series       = {Lecture Notes in Computer Science},
  volume       = {16806},
  pages        = {3--35},
  publisher    = {Springer},
  year         = {2026},
  url          = {https://doi.org/10.1007/978-3-032-35415-0\_1},
  doi          = {10.1007/978-3-032-35415-0\_1}
}

@misc{cryptoeprint:2026/1164,
      author = {Zirui Chen and Shi Tang and Zhengchao Gao and Yongjia Su and Lingyue Qin and Xiaoyang Dong},
      title = {Algebraic Cryptanalytic Extraction on Hard-Label Neural Networks},
      howpublished = {Cryptology {ePrint} Archive, Paper 2026/1164},
      year = {2026},
      url = {https://eprint.iacr.org/2026/1164}
}

@article{QLW26,
  author       = {Xiaokang Qi and Hao Lei and Longxiang Wei and Xiaohan Sun and Meiqin Wang},
  title        = {Cryptanalytic Extraction of Neural Networks with Various Activation Functions},
  journal      = {IACR Transactions on Symmetric Cryptology ({ToSC})},
  volume       = {2026},
  number       = {1},
  pages        = {468--505},
  year         = {2026},
  doi          = {10.46586/tosc.v2026.i1.468-505},
  url          = {https://doi.org/10.46586/tosc.v2026.i1.468-505}
}

@inproceedings{ADF26,
  author       = {Roderick Asselineau and
                  Patrick Derbez and
                  Pierre-Alain Fouque and
                  Brice Minaud},
  editor       = {Nadia Heninger and
                  Mike Rosulek},
  title        = {Cryptanalytic Extraction of Deep Neural Networks with Non-linear
                  Activations},
  booktitle    = {Advances in Cryptology - {CRYPTO} 2026 - 46th Annual International
                  Cryptology Conference, Santa Barbara, CA, USA, August 17-20, 2026,
                  Proceedings, Part {VII}},
  series       = {Lecture Notes in Computer Science},
  volume       = {16806},
  pages        = {36--66},
  publisher    = {Springer},
  year         = {2026},
  url          = {https://doi.org/10.1007/978-3-032-35415-0\_2},
  doi          = {10.1007/978-3-032-35415-0\_2}
}

@misc{cryptoeprint:2024/1580,
      author = {Nicholas Carlini and Jorge Ch{\'{a}}vez{-}Saab and Anna Hambitzer and
                Francisco Rodr{\'{\i}}guez{-}Henr{\'{\i}}quez and Adi Shamir},
      title = {Polynomial Time Cryptanalytic Extraction of Deep Neural Networks in the
               Hard-Label Setting (Extended Version)},
      howpublished = {Cryptology {ePrint} Archive, Paper 2024/1580},
      year = {2024},
      url = {https://eprint.iacr.org/2024/1580}
}

@misc{shen2026fcn,
      author = {Yantian Shen and Yi Chen and Anyu Wang and Hongbo Yu and Xiaoyun Wang},
      title = {Cryptanalytic Extraction of Neural Networks Without Known Architecture
               Assumption},
      howpublished = {arXiv:2609.14379 [cs.CR]},
      year = {2026},
      doi = {10.48550/arXiv.2609.14379},
      url = {https://arxiv.org/abs/2609.14379}
}

@misc{cryptoeprint:2026/1943,
      author = {Duo Xu and Liu Zhang and Zilong Wang},
      title = {Finite-Precision Error Analysis of Cryptanalytic Model Extraction},
      howpublished = {Cryptology {ePrint} Archive, Paper 2026/1943},
      year = {2026},
      url = {https://eprint.iacr.org/2026/1943}
}

\clearpage

% For citations of references, we prefer the use of square brackets
% and consecutive numbers. Citations using labels or the author/year
% convention are also acceptable. The following bibliography provides
% a sample reference list with entries for journal
% articles~\cite{ref_article1}, an LNCS chapter~\cite{ref_lncs1}, a
% book~\cite{ref_book1}, proceedings without editors~\cite{ref_proc1},
% and a homepage~\cite{ref_url1}. Multiple citations are grouped
% \cite{ref_article1,ref_lncs1,ref_book1},
% \cite{ref_article1,ref_book1,ref_proc1,ref_url1}.
%
\end{document}